%% file: bce_econometrics_auctions.tex
\documentclass[12pt,letterpaper]{article}
\pdfoutput=1

\usepackage{cleveref}
\usepackage[margin=1.23in]{geometry}
\usepackage{amsmath,amsfonts,graphpap,amscd,mathrsfs,graphicx,lscape}
\usepackage{subcaption}
\usepackage{amssymb,amstext,xspace}
\usepackage{float}	
\usepackage{hyperref}
\usepackage{bm}
\usepackage{bbm}
\usepackage{natbib}

\usepackage{color}              %

\usepackage[suppress]{color-edits}
\addauthor{et}{blue}
\addauthor{vs}{red}
\addauthor{jz}{cyan}

\usepackage{amsthm}

\newcommand{\E}{\mathbb{E}}

\newtheorem{theorem}{Theorem}
\newtheorem{result}[theorem]{Result}
\newtheorem{lemma}[theorem]{Lemma}

\newtheorem{corollary}[theorem]{Corollary}

\newtheorem{remark}{Remark}
\newtheorem{defn}{Definition}

\def\squareforqed{\hbox{\rule{2.5mm}{2.5mm}}}

\def\QED{\ifmmode\squareforqed %
  \else{\nobreak\hfil   %
    \penalty50                 %
    \hskip1em                  %
    \null                      %
    \nobreak                   %
    \hfil                      %
    \squareforqed              %
    \parfillskip=0pt           %
    \finalhyphendemerits=0     %
    \endgraf}                  %
  \fi}

\def\blksquare{\rule{2mm}{2mm}}
\def\qedsymbol{\blksquare}
\newcommand{\bg}[1]{\medskip\noindent{\bf #1}}
\newcommand{\ed}{{\hfill\qedsymbol}\medskip}

\newenvironment{example}{\bg{Example. }}{\ed}

\newcommand{\R}{\ensuremath{\mathbb R}}

\newcommand{\comment}[1]{}
 {}%

\newcommand{\junk}[1]{}

\newlength{\tmp} \newlength{\lpsx} \newlength{\lpsy} \newlength{\upsx} \newlength{\upsy}

\newcommand{\Var}{\ensuremath{\text{Var}}}

\newcommand{\Omit}[1]{}

\renewcommand{\vec}[1]{\ensuremath{{\bf #1}}}

\newcommand{\cS}{\ensuremath{{\cal S}}}

\begin{document}
\title{Inference on Auctions with Weak Assumptions on Information\footnote{We thank participants at the Econometric Society China Meeting in Wuhan - 2017 and at the R. Porter Festschrift at Northwestern University  in May 2017 for comments. Code for monte-carlo simulations and ocs data analysis using our methods can be found at \url{https://github.com/vsyrgkanis/information_robust_econometrics_auctions}}}
\author{
Vasilis Syrgkanis\footnote{Microsoft Research, New England; \href{mailto:vasy@microsoft.com}{vasy@microsoft.com}} 
\and 
Elie Tamer\footnote{Department of Economics, Harvard University; \href{mailto:elietamer@fas.harvard.edu}{elietamer@fas.harvard.edu}}
\and
Juba Ziani\footnote{Department of Computing and Mathematical Sciences, California Institute of Technology; \href{mailto:jziani@caltech.edu}{jziani@caltech.edu}. Supported by NSF grants CNS-1331343 and CNS-1518941, and the US-Israel Binational Science Foundation grant 20122348}
}
\date{\today}

\maketitle    

\begin{abstract}
Given a sample of bids from independent auctions, this paper examines the question of inference on auction fundamentals (e.g. valuation distributions, welfare measures) under weak assumptions on information structure. The question is important as it allows us to learn about the valuation distribution in a robust way, i.e., without assuming that a particular information structure holds across observations. We leverage the recent contributions of \cite{Bergemann2013} in the robust mechanism design literature that exploit the link between Bayesian Correlated Equilibria and Bayesian Nash Equilibria in incomplete information games to construct an econometrics framework for learning about auction fundamentals using observed data on bids. We showcase our construction of identified sets in private value and common value auctions. Our approach for constructing these  sets inherits the computational simplicity of solving for correlated equilibria: checking whether a particular valuation distribution belongs to the identified set is as simple as determining whether a {\it linear} program is feasible. A similar linear program can be used to construct the identified set on various welfare measures and counterfactual objects. For inference and to summarize statistical uncertainty, we propose novel finite sample methods using tail inequalities that are used to construct confidence regions on sets. We also highlight methods based on  Bayesian bootstrap and subsampling. A set of Monte Carlo experiments show adequate finite sample properties of our inference procedures. We illustrate our methods using data from OCS auctions. 

\end{abstract}
\baselineskip=1.5\baselineskip
\newpage
\section{Introduction}
\input{intro.tex}

\section{Bayes-Correlated Equilibria and Information Structure Uncertainty}\label{sec:bce_inference}
\input{bce_inference.tex}

\section{Common Value Auctions}\label{sec:common-value}
\input{common-value.tex}

\section{Private Value Auctions}\label{sec:private-value}
\input{private-value.tex}

\section{Inference on Sharp Sets in CV Auctions}\label{sec: estimation}\label{sec:inference}
\input{inference.tex}

\section{Monte-Carlo Analysis: Common Value Auctions}\label{sec: simulations}\label{sec:monte-carlo}
\input{monte-carlo.tex}

\section{Empirical Illustration: OCS Auctions}\label{sec:ocs}
\input{ocs-data-analysis.tex}

\section{Conclusion}\label{sec:conclusion}

We provide a framework for inference on auction fundamentals without making strong restrictions on information. Using data, we use the recent results in theory to characterize the identified set for these primitives by exploiting the linear programming structure of the set of Bayesian correlated equilibria. We have several applications of the approach mainly to common value and private value auctions and other scenarios. Our results can also be used by mechanism designers in that the data allows us to restrict the domain of signal/state of the world distribution, which would lead to sharper mechanisms. We also provide approaches to inference, and propose finite sample approaches to building confidence regions for sets in partially identified models. 

\bibliographystyle{abbrvnat}
\bibliography{bib-AGT}

\appendix

\input{appendix.tex}

\end{document}

%% file: intro.tex
A recent literature in robust mechanism design studies the following question: given a game \vsedit{(e.g. an auction)}, what are the possible outcomes - such as welfare or revenue - that arise under different information structures? This literature is motivated by {\it robustness}, i.e., characterizing outcomes that can occur in a given game under weak assumptions on information (See \cite{Bergemann2013}). For example, in auction models, in addition to specifying the details of the game in terms of bidder utility function, and bidding rules, one needs to specify the information structure (what players know about the state of the world and the information possessed by other players) to be able to derive the Nash equilibrium. 

In particular, \vsedit{%
in the Independent Private Values (IPV)} setting, players know their \vsedit{value for the item}, which is assumed independent from other player \vsedit{values}, and receive no further information about their opponents' values. \vsedit{The latter typically yields a unique equilibrium outcome. However, different assumptions on what signals players have about opponent values prior to bidding in the auction, lead to different equilibrium outcomes.} \vsedit{Auction data} rarely contain information on what bidders knew and what their information sets included, \etedit{and given that this information leads to different outcomes}, it would be interesting to analyze what can happen when we relax the independence assumption in such auctions \etedit{by allowing bidders to know some information about their opponents' valuations}. \cite*{Bergemann2016c} (BBM) examine exactly this question in an auction game and provide achievable bounds on various outcomes\vsedit{, such as the revenue of the auction as a function of auction fundamentals like the distribution of the common value}. 

In this paper, we address  the following econometrics question, \etedit{which is the reverse of the one posed by BBM above}: given an i.i.d. sample of auction data (independent copies of bids from a set of auctions), what can we learn about auction fundamentals, such as the distribution of values, when we make weak assumptions on the information structure? We maintain throughout that players play according to a Bayesian Nash equilibrium (BNE) but allow these bidders to have different information structures \etedit{in different auctions}\vsedit{, i.e. receive different types of signals prior to bidding.}
In particular, we use these observations (bids and other observables) from a set of independent auctions to construct {\it sets} of valuation distributions that are consistent with both the observed distribution of bids and the \vsedit{known auction rules} %
maintaining that players are Bayesian. We exploit the robust predictions in a given auction \`{a} la BBM to conduct econometrically robust predictions of  auction fundamentals given a set of data; i.e., robust economic prediction leads to robust inference. 

\ \ 

Key to our approach is the characterization of sharp sets of valuation distributions (and other functionals of interest) via computationally attractive procedures. This is a result of the equivalence between a particular class of {\it Bayes Correlated Equilbria} (or BCE) and  Bayes Nash Equilibria (or BNE) for a similar game with an arbitrary information structure. It is well known that BCE can be computed efficiently since they are solutions to linear programs (as opposed to BNE which are hard to compute). \vsedit{Moreover, exploiting a result of \cite{Bergemann2016} (see also \cite{aumann1987}), we show that there is an equivalence between the set of fundamentals that obey the BCE restrictions and the fundamentals that obey the BNE constraints under {\it some} information structure.} %
This equivalence is the key to our econometrics approach. The formal statistics program that ensues is one where the sharp set satisfies a set of linear equality and inequality constraints. If we knew the true distribution of bids, then identifying the parameters would be a simple computational problem of solving a linear program. We do not observe the true bid distribution, but this distribution can be estimated consistently from the observed data. Thus, we use the estimated bid distribution to solve for an estimate of the identified set. We are also able to characterize sampling uncertainty to obtain various notions of confidence regions \etedit{covering the identified with a prespecified probability}. 

In addition to learning about auction primitives, we show how our approach using data on auctions can be used to construct identified sets for auction welfare measures, seller surplus and other objects. Information on auction primitives, along with these measures obtained using our procedures that combine data with the theory, can be used to guide future market designers to better study particular auction setups (or use our results in other markets). 

 \vsedit{Importantly, we address the problem of counterfactual estimation: what would  the revenue or surplus have been had we changed the auction rules? We formulate notions of informationally robust counterfactual analysis and we show that such counterfactual questions can also be phrased as solutions to a single linear program, \etdelete{that} simultaneously captur\etdelete{es}\etedit{ing} equilibrium constraints in the current auction as well as the new target auction. We show that even without recovering the information structure from the data, an analyst can perform robust counterfactual analysis and answer the following question: under an arbitrary information structure in the current auction which produced the  data at hand,  what is the best and worst value of a given quality measure (e.g. welfare, revenue) in the new target auction under an arbitrary information structure? Thus we can get estimates of the upper and lower bounds of a given quantity in the new auction design, in a way that is robust to information structure and without the need of recovering it.} \etedit{Hence, this approach to inference requires minimal informational assumptions on the data generating process (DGP) that generated the existing data set and minimal informational assumptions on the counterfactual auction that a market designer is contemplating to run. }

\ 

The closest work to our paper is \citealt*{Bergemann2016c}, where the authors provide worst-case bounds on the revenue of a common value first price auction as a function of the distribution of values. Their approach does not use the bid distribution as input, unlike our approach which obtains an estimate of the bid distribution from the data.  The main approach in BBM is to show that the revenue cannot be too small since at BCE, no player wants to deviate to any other action and so players do not want to deviate to a specific type of a deviation which is the following: conditional on your bid, deviate uniformly at random above your bid (upwards deviation). Hence, the bound on the mean they provide uses a {\it subset} of the set of best response deviations that are allowed \vsedit{so as to bound the bid of a player as a function of his value. \etdelete{and} \etedit{In drawing a connection between the equilibrium bid and the value,} \etedit{this bound} is by definition loose \etdelete{in drawing a connection between the equilibrium bid and the value} (and can be very loose - bound twice as large as identified set - as we show in an example in Appendix \ref{appendix A}).} Given data, we are able to learn the bid distribution and hence are not constrained to look at only these bid-distribution-oblivious upwards deviations. We can instead compute an optimal \vsedit{deviating} bid for this given bid distribution and use the constraint that the player does not want to deviate to this \vsedit{distribution-tailored action.} %
This allows us to bound the unobserved value of the player as a function of the observed bid \etdelete{This approach} leading to a {\it sharp characterization} of \vsedit{auction fundamentals} using the data. Also, the approach taken to inference in this paper is deliberately conservative in that we try to make weak or no assumptions on information while maintaining Nash behavior. This is in the same spirit as \cite{HaileTamer} who study the question of inference in English auctions under minimal assumptions and derive estimable bounds on the distribution of bidder valuations. 

Another paper that uses a similar insight of studying the econometrics of games with weak information is the recent work of \cite{Magnolfi2016} on inference in entry game models. The approach used there, though similar in motivation, does not transfer easily to studying general auction mechanisms. 

The paper is organized as follows. Section \ref{sec:bce_inference} introduces the problem and provides formal definitions of the objects of interest. We then state our identification results given an i.i.d. set of data on bids. This identification is constructed via a linear program where we show how various constraints (such as symmetry, parametric restrictions, etc) can be incorporated. \vsedit{We also show how computing sharp sets for the expected value of moments of the fundamentals amounts to solving two linear programs and how robust counterfactual analysis of some metric function, with respect to changes in the auction, can also be handled in a computationally efficient manner.} \vsedit{We then provide two example applications of the general setup: one for common value auctions (Section \ref{sec:common-value}) and another for \vsedit{private value} auctions (Section \ref{sec:private-value}). Section \ref{sec:inference} provides our estimation approach for constructing confidence intervals on the estimated quantities from sampled datasets using sub-sampling methods and finite sample concentration inequality approaches.} Section \ref{sec:monte-carlo}  examines the finite sample performance of the large linear program using a set of Monte Carlo simulations. These show adequate performance in IPV and Common Value (CV) setups. Section \ref{sec:ocs} illustrates our inference approach using auction data from OCS wildcat oil auctions and show how the statistical algorithm can be used to derive bounds on valuation distributions. \vsdelete{Section ? provides further extensions and Section ? concludes.} \etedit{Finally, the Appendix contains results on the sharpness of the BBM bounds, and bounds on the mean of the valuation in common value auctions with different smoothness assumptions. }

%% file: bce_inference.tex
We consider a game of incomplete information among $n$ players. There is an unknown payoff-relevant state of the world $\theta\in \Theta$.  This state of the world enters directly in each player's utility. Each player $i$ can pick from among a set of actions 
$A_i$ and receives utility which is a function of the payoff-relevant state of the world $\theta$ and the action profile of the players $\vec{a}\in A\equiv A_1\times\ldots\times A_n$: $u_i(\vec{a};\theta)$. This along with a prior on $\theta$ (defined below) will represent the game structure that we denote by $G$, as separate from the information structure which we will define next.

 Conditional on the state of the world each player receives some minimal signal $t_i\in T_i$. The state of the world $\theta\in \Theta$ and the vector of signals $\vec{t}\in T\equiv T_1\times\ldots\times T_n$ are drawn from some joint measure\footnote{The setup is general in that the set $T$ is unrestricted.} $\pi \in \Pi\subseteq \Delta(\Theta\times T)$. The signals $t_i$'s can be arbitrarily correlated with the state of the world. We denote such signal structure with $S$. This defines the game $(G,S).$\footnote{\vsedit{We will assume throughout that the set of $n$ players is fixed and known. If in the data the set of participating players varies across auctions \etedit{and players do not necessarily know the number or bidders}, then we can consider the superset of all players and simply assign a special type to each non-participating player. Conditional on this type, players always choose a default action (e.g. bidding zero in an auction). Then dependent on whether players observe the number of participants before submitting an action can be encoded by whether the players receive as part of their default signal, whether some player's type is the special non-participating type. In our auction applications we will assume that players do not necessarily observe the entrants in the auction before bidding.}} 

We consider a setting where prior to picking an action each player receives some additional information in the form of an extra signal $t_i'\in T_i'$. The signal vector $t'=(t_1',\ldots,t_n')$ can be arbitrarily correlated with the true state of the world and with the original signal vector $t=(t_1,\ldots,t_n)$. We denote  such augmenting signal structure with $S'$ and  the set of all possible such augmenting signal structures with $\cS'$. This will define a game $(G, S')$. Subsequent to observing the signals $t_i$ and $t_i'$, the player picks an action $a_i$. A Bayes-Nash equilibrium or BNE in this game $(G, S')$ is a mapping $\sigma_i:T_i\times T_i'\rightarrow \Delta(A_i)$ from the pair of signals $t_i,t_i'$ to a distribution over actions, for each player $i$, such that each player maximizes his expected utility conditional on the signals s/he received. 

A fundamental result in the literature on robust predictions (see \cite{Bergemann2013,Bergemann2016}), is that the set of joint distributions of outcomes $a\in A$, unknown states $\theta$ and signals $t$, that can arise as a  BNE of incomplete information under an arbitrary additional information structure in $(G, S')$, is equivalent to the set of Bayes-Correlated Equilibria, or BCE in $(G, S)$. \etedit{So, every BCE in $(G,S)$ is a BNE in $(G,S')$ for some augmenting information structure $S'$.} \etdelete{there is an augmenting information structure $S'$ such that}  We give the formal definition of BCE next. 

\ \ \

\begin{defn}[Bayes-Correlated Equilibrium] A joint distribution $\psi\in \Delta(\Theta\times T\times A)$ is a Bayes-correlated equilibrium of $(G,S)$ if for each player $i$, signal $t_i$, action $a_i$ and deviating action $a_i'$:
\begin{equation}
\E_{(\theta,t,a)\sim \psi}\left[u_i(a;\theta)~|~a_i,t_i\right]\geq \E_{(\theta,t,a)\sim \psi}\left[u_i(a_i',a_{-i};\theta)~|~a_i,t_i\right]
\end{equation} 
and such that the marginals with respect to signals and payoff states are preserved, i.e.:
\begin{equation}
\forall \theta \in \Theta, t\in T: \sum_{a\in A} \psi(\theta,t,a) = \pi(\theta,t)
\end{equation}
\end{defn}

\  \
An equivalent and simpler way of phrasing the Bayes-correlated equilibrium conditions is that:
\begin{equation}
\forall t_i,a_i,a_i': \sum_{\theta,t_{-i},a_{-i}} \psi(\theta,t,a)\cdot \left(u_i(a;\theta)-u_i(a_i',a_{-i};\theta)\right) \geq 0
\end{equation}
We state the main result in \cite{Bergemann2016} next. 
\ \ \

\begin{theorem}[\cite{Bergemann2016}] A distribution $\psi\in \Delta(\Theta\times T\times A)$ can arise as the outcome of a Bayes-Nash equilibrium under some augmenting information structure $S'\in \cS'$, if and only if it is a Bayes-correlated equilibrium in $(G,S)$.
\end{theorem}

The robustness property of this result is as follows. The set of BNE for $(G,S')$ (think of an auction with unknown information) is the same as the set of BCE for $(G,S)$ where $S'$ is an augmented information structure derived from $S$. So, we will not need to know what is in $S'$, but rather we could compute the set of BCE for $(G,S)$ and the Theorem shows that for each BCE, there exists a corresponding information structure $S'$ in $\cS'$ and a BNE of the game $(G,\cS')$ that implements the same outcome.

\subsection{The Econometric Inference Question}

We consider the question of inference on auction fundamentals using data under weak assumptions on information. In particular, assume we are given sample of observations of action profiles $a^1,\ldots,a^N$ from an incomplete information game $G$. Assume that we do not know  the exact augmenting signal structure $S^1,\ldots,S^N$ that occurred in each of these samples where it is implicitly maintained that $S^i$ can be different from $S^j$ for $i \neq j$, i.e.,  the signal structure in the population is drawn from some unknown mixture.  Also, maintaining that players play Nash, or that  $a^t$ was the outcome of some Bayes-Nash equilibrium or BNE under signal structure $S^t$,  We study the question of inference on  the distribution $\pi$ of the fundamentals of the game. Under the maintained assumption,  we characterize the sharp set of possible distributions of fundamentals $\pi$ that could have generated the data. This allows for policy analysis within the model without making strong restrictions on information. 

A similar question was recently analyzed in the context of entry games by \cite{Magnolfi2016}, where the goal was the identification of the single parameter of interaction when both players choose to enter a market. In this work we ask this question in an auction setting and attempt to identify the distribution of the unknown valuations non-parametrically.

The key question for our approach is to allow for observations on different auctions to use different (and unobserved to the econometrician) information structures, and, given the information structures, that different markets or observations on auctions, to use a different BNE. 
Given this equivalence of the set of Bayes-Nash equilibria under some information structure and the set of Bayes-correlated equilibria, and given that the set of BCE is convex, allowing for this kind of heterogeneity is possible.  Heuristically, given a distribution over action profiles $\phi\in \Delta(A)$ (which is constructed using the data), there exists a mixture of information structures and equilibria under which $\phi$ was the outcome. The process by which we arrived at $\phi$ is by first picking an information structure from this mixture and then selecting one of the  Bayes-Nash equilibria for this information structure. This is possible if and only if there exists a  distribution $\psi(\theta,t,a)$, that is a Bayes-correlated equilibrium and such that $\sum_{\theta,t}\psi(\theta,t,a)=\phi(a)$ for all $a\in A$. 
Again, we start with the elementary information structure $S$ and maintain that observations in the data are expansions of this information structure. So, for a given market $i$, any BNE using information structure $S^i$ is a BCE under $S$. The data distribution of action is a mixture of such BNE over various information structures and hence it would map into a mixture of BCE under the same $S$. Since the set of BCE under $S$ is convex, any mixtures of elements in the set is also a BCE. So then intuitively, the set of primitives that are consistent with the model and the data is the set of BCEs, $\psi(\theta,t,a)$ such that $\sum_{\theta,t}\psi(\theta,t,a)=\phi(a)$ for all $a\in A$. \newline To conclude, the convexity of the set of BCEs allows us to relate a distribution of bids from an iid sample to a mixture of BCEs. This is possible since the distribution of bids  uses a mixture of signal structures, which essentially coincides with a mixture of BCEs, itself another BCE by convexity. We summarize this discussion with a formal result.

\ \ \

\begin{lemma}\label{lemma1}
Consider a model which conditional on the unobservables, yields a convex set of possible predictions on the observables. Then the sharp identified set of the unobservables under the assumption that exactly one of these predictions is selected in our dataset, is identical to the sharp identified set under the assumption that our dataset is a mixture of selections from these predictions.
\end{lemma} 
\ \
\begin{proof}
Suppose that what we observe in the data is a convex combination of feasible predictions of our model. Then by convexity of the prediction set, this convex combination is yet another feasible prediction of our model. Hence, this is equivalent to the assumption that this single prediction is selected in our dataset. 
\end{proof}
\ \ \

Finally, we provide next the main engine that allows for construction of the observationally equivalent set of primitives that obey model assumptions and result in a distribution on the observables that match that with the data. We state this as a Result. 

\ \

\begin{result}\label{res:bce_inference} Let there be a distribution $\phi$ defined on the space of action profiles. Given the setup and results above,  the set of feasible joint distributions of signals and types $\pi$ that are consistent with  $\phi$ are the set of distributions for which the following linear program is feasible:
\begin{align}
LP(\phi,\pi) ~~~&~~~ \forall t_i,a_i,a_i':~& & \sum_{\theta,t_{-i},a_{-i}} \phi(a)\cdot x(\theta, t|a)\cdot \left(u_i(a;\theta)-u_i(a_i',a_{-i};\theta)\right) \geq 0\\
~~~&~~~ \forall (\theta,t):~& & \pi(\theta,t) = \sum_{a\in A} \phi(a)\cdot x(\theta,t|a)\\
~~~&~~~ & & x(\cdot,\cdot | \alpha) \in \Delta(\Theta\times T) 
\end{align}
where $\sum_{\theta,t} \psi(\theta,t,\alpha)= \phi(\alpha)$ for all $\alpha\in A$, and  $\psi(\theta,t,\alpha) = \phi(\alpha)x(\theta,t|\alpha)$.

\vskip .1in
\noindent Equivalently, the sharp set for the distribution $\pi\in \Delta(\Theta\times T)$:
\begin{align*}
\Pi_I(\phi) =\{ \pi \in \Pi: LP(\phi,\pi) \text{ is feasible }\}
\end{align*}

\end{result}

\ \ \

\etedit{The above result is generic, in that it handles general games with generic states of the world $\theta.$ In particular, it nests both standard private and common value auction models and provides a mapping between the distribution of bids $\phi$ and the set of feasible distribtions over signals and $\theta.$}
An iid assumption on bids along with a large sample assumption allow us to learn the function $\phi(.)$ (asymptotically). So, given the data, we can consistently estimate $\phi$. This is a maintained assumption that we require throughout. Given $\phi$, the above result tells us how to map the estimate of $\phi$ to the set of BNE that are consistent with the data and are robust to any information structure that is an expansion of a minimal information structure $S$. Suppose we assume that both $t$ and $\theta$ take finitely many values (an assumption we maintain throughout), then a joint distribution on $(t,\theta)$, $\pi(t,\theta)$ is consistent with the model and the data if and only if it solves the above {\it linear program}. The LP formulation is general, but in particular applications, it is possible to use parametric distributions for the $\pi$. In addition, it is possible for the above LP to allow for observed heterogeneity by using covariate information whereby this LP can be solved accordingly \vsedit{(see an example such adaptation in the common value Section \ref{sec:common-value}).}  
 Though the above LP holds in general (and covers both common and private values for example), we specialize in the next Sections the above LP to standard cases studied in the auction literature, mainly common values an private values models.

\vscomment{I think we need to move the corollary that we can calculate the identified set of the expected value of any function of $\theta,t$, by simply solving two linear programs, which will yield the upper and lower bounds of this ``projection'' of the identified set. This will make it way more general in this general setup. I also think that we might want to move the counter-factuals up here, as they can be phrased in the general setup and don't need to be a common value or private value auction. I implemented these two sections. Let me know what you think. \etedit{I AGREE!}
} 

{\subsection{Identified Sets of Moments of Fundamentals}
In the case of non-parametric inference where we put no constraint on the distribution of fundamentals, i.e. $\Pi=\Delta(\Theta\times T)$,  observe that the sharp set is linear in the density function of the fundamentals $\pi(\cdot,\cdot)$. Therefore, maximizing or minimizing any linear function of this density \etdelete{variables} can be performed via solving a single linear program. This implies that we can evaluate the upper and lower bounds of the expected value of any function $f(\theta, t)$ of these fundamentals, in expectation over the true underlying distribution. The latter holds, since the expectation of any function with respect to the underlying distribution is a linear function of the density. \etedit{We state this as a corollary next.}
\begin{corollary}
Let Result (\ref{res:bce_inference}) hold. Also, let $f(.,.)$ be any function of the state of the world and the profile of minimal signals (eg, $f(\theta, t)=\theta$ or $f(\theta, t) = \theta^2$, ...). The sharp identified set for the expected value of $f(.,.)$ w.r.t. the true distribution of fundamentals, i.e. $E_{(\theta,t) \sim \pi}[f(\theta,t)]$ is an interval $[L,U]$ such that:
\begin{align}
L =~& \min_{\pi \in \Pi_I(\phi)} \sum_{\theta \in \Theta, t\in T} f(\theta, t)\cdot \pi(\theta, t)\\
U =~& \max_{\pi \in \Pi_I(\phi)} \sum_{\theta \in \Theta, t\in T} f(\theta, t)\cdot \pi(\theta, t)
\end{align}
\etedit{Without further assumptions on $\Pi$, i.e., when} \etdelete{Under non-parametric inference, i.e.} $\Pi=\Delta(\Theta,T)$, these are two linear programming problems.
\end{corollary}
}

\subsection{Robust Counterfactual Analysis}\label{sec:counterfactual}
\input{counterfactual-generic.tex}

%% file: counterfactual-generic.tex
\vsedit{Suppose that we wanted to understand the performance of some other auction when deployed in the same market, with respect to some \etedit{objective or} metric: $F:\Theta\times T\times A\rightarrow \R$ that is a function of the unknown fundamentals and the action vector. Examples of such metrics in single-item auctions could be social welfare: $F(\theta, t, a)=\sum_{i=1}^n \theta_i x_i(a)$ or revenue, i.e. $F(\theta, t, a)=\sum_{i=1}^n p_i(a)$, where $\theta_i$ is the value of player $i$ for the item at sale, $x_i(a)$ is the probability of allocating to player $i$ under action profile $a$ and $p(a)$ is the expected payment of player $i$ under action profile $a$.}

\vsedit{We are interested in computing an upper and lower bound on this metric under this new auction which has different utilities $\tilde{u}_i(a;\theta)$ and under any Bayes-correlated equilibrium which would map into a BNE with an augmenting information structures. This is important since it allows us to obtain bounds on welfare or other metrics in a new environment \etdelete{under general information structures}. The welfare bounds computed in this manner will inherit the robustness property in that they will be valid under all information structures. }

\vsedit{This is straightforward in our setup since computing a sharp identified set for any such counter-factual can be done in a computationally efficient manner, in both the common value setting and in the correlated private value setting. The upper bound of the counter-factual \etdelete{boils down to} \etedit{can be obtained} using the following linear program, that takes as input the observed distribution of bids in our current auction, the metric $F$, and the primitive utility form $\tilde{u}=(\tilde{u}_1,\ldots,\tilde{u}_n)$ under the alternative auction. We state this result in the next Theorem.}

\vsedit{\begin{theorem}
Given a metric function $F:\Theta\times B\rightarrow \R$, a distribution over action profiles $\phi()$ and vector of alternative auction utilities $\tilde u$, the sharp upper and lower bounds on the expected metric under the new auction can be computed using the following LP:
 \begin{align*}
 LP(\phi,F,\tilde{u}) ~~~&~~~ \underset{\tilde{\psi}}{\min/\max}\sum_{\theta, t, a} F(\theta, t, a) \cdot \tilde{\psi}(\theta, t, a)  \\ 
 ~~~&~~~ \forall t_i, a_i, a_i':  \sum_{\theta,t_{-i},a_{-i}} \tilde{\psi}(\theta, t, a)\cdot \left(\tilde{u}_i(a;\theta)-\tilde{u}_i(a_i',a_{-i};\theta)\right) \geq 0\\
 ~~~&~~~ \forall (\theta, t)\in \Theta\times T:~  \pi(\theta, t) = \sum_{a\in A} \tilde{\psi}(\theta, t, a)\\
 ~~~&~~~  \tilde{\psi} \in \Delta(\Theta\times T\times A) \text{ and } \pi \in \Pi_I(\phi)
 \end{align*}
 In the non-parametric case, where $\Pi=\Delta(\Theta\times T)$, the latter is a linear program.
 \end{theorem}}
 
 \etedit{Note that getting sharp bounds on welfare measures for example using the above Theorem does not require one to infer in a prior step the distribution over the primitives. Rather, the above procedure provides sharp bounds on this welfare measure using a linear program.}

%% file: common-value.tex
\etedit{We specialize the above results to important classes of auctions.}
We begin with the common value model  in which the game of incomplete information $G$ is a single item common value auction. In this case the unknown state of the world is the unknown common value of the object $v$, which we assume to take values in some finite set $V$. Moreover, we initially assume that the minimal information structure is degenerate, i.e., players receive no minimal signal about this unknown common value\footnote{Other constraints on the initial signals are allowed and here we take the degenerate signal for simplicity.}. The signal set $T$ becomes a singleton and is irrelevant. Thus we will denote with $\pi\in \Delta(V)$ the distribution of the unknown common value, which is the parameter that we wish to identify. This is a particularly simple model to illustrate the structure of the LP approach and showcase the flexibility of our methods. So, in this particular model, we want to learn the distribution of the state of the world which is the common valuation distribution. 

Prior to bidding in the auction, the players receive some signal which is drawn from some distribution; this signal can be correlated with the unknown common value and with the signals of his opponents. We wish to be ignorant about which information structure realized in each auction sample and want to identify the sharp identified set for $\pi$. Moreover, we will assume that the players' bids take values in some discrete set $B$ and players play a BNE. The characterization of the identified set for $\pi$ in this model is stated in the Theorem below. 

\ \ 

\begin{theorem}\label{CommonValue Result}

Let the common value model above hold with $n$ bidders. Given a distribution of bid profiles $\phi$ supported on a set $S\subseteq B^{n}$, the set of distributions of bids $\pi\in \Delta(V)$ that are consistent with $\phi$ are ones where the following  program is feasible: 
\begin{align*}
LP(\phi,\pi) ~~~&~~~ \forall b_i^*,b_i'\in B:~& & \sum_{v\in V,~ \vec{b}\in S: b_i = b_i^*} \phi(\vec{b}) \cdot x(v | \vec{b})\cdot \left(u_i(\vec{b};v)-u_i(b_i',\vec{b}_{-i};v)\right) \geq 0\\
~~~&~~~ \forall v\in V:~& & \pi(v) = \sum_{\vec{b}\in S} \phi(\vec{b})\cdot x(v|\vec{b})\\
~~~&~~~ \forall \vec{b}\in S:~& & x(\cdot|\vec{b}) \in \Delta(V)
\end{align*}
\end{theorem}

\ \ \

Observe, that in this setting, the latter linear program is also linear in $\pi$. Thus we get that the sharp identified set is a convex set and is defined as the set of solutions to the above linear program, where $\pi$ is also a variable. 

This observation also allows us to easily infer upper and lower bounds on any linear function of the unknown distribution $\pi$. This is stated next as a Corollary to the above Theorem.

\ \ \

\begin{corollary}
Let Result (\ref{CommonValue Result}) hold. Also, let $f(.)$ be any function of valuations (such as $f(v)=v$ or $f(v) = v^2$, ...). Then, using the LP above, we can get upper and lower bounds on such moments as follows:
\begin{align}
\max_{\pi \in \Pi_I(\phi)} \sum_{v\in V} f(v)\cdot \pi(v)\\
\min_{\pi \in \Pi_I(\phi)} \sum_{v\in V} f(v)\cdot \pi(v)
\end{align}
\end{corollary}
Observe that the latter linear expressions are simply: $\E_{v\sim \pi}[f(v)]$ and so the above shows that we can  compute in polynomial time upper and lower bounds of any moment of the unknown distribution of the common value.

Also, note that the Corollary above shows that to do set inference with respect to any moment of the unknown distribution, we do not need to discretize the space of probability distributions and enumerate over all probability vectors, checking whether they are inside the sharp identified set. Rather we can just solve the above LP.

Essentially this observation says that we can easily compute the support function $h(z; \Pi_I(\phi))$ of the identified set $\Pi_I(\phi)$ at any direction $z$, by simply solving a linear program. 
In Section (\ref{upper bound}), we provide an upper bound on the mean valuation distribution when the latter is continuous. This upper bound is derived in terms of the observed bids distribution. 

\ \ \

\begin{remark}[\bf Winning bid] The above procedure can be easily modified if indeed as it may be the case, only winning bids are observed. In particular, given that I observe the CDF $F(.)$ of the winning bid, I know that an equilibrium is consistent with said CDF if and only if for every possible bid in $B$,
	$$
	F(x)= \sum\limits_{v \in V} \sum\limits_{\vec{b} / b_i \leq x \; \forall i} \psi(\vec{v},\vec{b})
	$$ 
	This has up to $B$ linear constraints of $B^{n+1}$ variables, but if we assume the bid vector distribution has small support of size $K$, we only need $K \cdot B$ variables.

\end{remark}

\ \ \

\begin{remark}[\bf Covariate Heterogeneity]
	Suppose we have covariates and want to allow for observed heterogeneity where we maintain the  assumption  that the vector of covariates $x$ takes finitely many values. Then, one nonparametric approach is to repeat and solve the above LP for $\pi(v|x=x_0)$ for every value $x_0$ that $x$ takes. In addition, in cases where we have $E[v|x] = x'\beta$, then we can solve directly for the identified  set for $\beta$ by solving the following LP, which we denote with $LP(\phi, \beta)$:
	
	\begin{align*}
	\forall b_i^*,b_i'\in B, x_0 \in X:~& & \sum_{v\in V,~ \vec{b}\in S: b_i = b_i^*} \phi(\vec{b}|x_0) \cdot x(v | \vec{b},x_0)\cdot \left(u_i(\vec{b};v)-u_i(b_i',\vec{b}_{-i};v)\right) \geq 0\\
	\forall x_0 \in X:~& & x_0'\beta = \sum_{\vec{b}\in S}\sum_v v  \phi(\vec{b}|x_0)\cdot x(v|\vec{b},x_0)\\
	 \forall \vec{b}\in S, x_0 \in X:~& & x(\cdot|\vec{b}, x_0) \in \Delta(V)
	\end{align*}
	Here, the $n\times k$ vector $x_0'\beta$ allows for different players to have different $\beta$'s (and the $x$'s in this case would be auction specific heterogeneity where the different $\beta$'s would allow the mean valuation of different players to depend differently on auction characteristics). 
	
\vsedit{The latter is attractive as it allows us to couple the identified set of the mean of the unknown common value across multiple covariate realizations, in a computationally tractable manner. Otherwise we would have to generate the identified sets of the mean for each covariate realization and then solve a second stage problem which would try to find the set of joint solutions in each of these identified sets (via some form of joint grid search) that are consistent with a model of how the conditional mean $E[v|x]$ varies as a function of $x$. However, joint grid search would grow exponentially with the number of realizations of the covariates. The latter remark, shows that when the model of $E[v|x]$ is linear, then we can save this exponential blow-up in the computation whilst leveraging the statistical power of coupling data from separate covariate realizations.} \vsdelete{The attractive }
	\end{remark}

\ \ \ \

%% file: private-value.tex
We now consider the case of a private value single item auction. In this case the (unknown) state of the world is the a vector of private values $\vec{v}=(v_1,\ldots,v_n)\in V^n$. We assume that these private values come from some unknown joint distribution $\pi\in \Delta(V^n)$. Moreover, we initially assume that players know at least their own private value. Thus  the (minimal) signal set $T_i$ is equal to $V$ and moreover, we have that conditional on a value vector $v$, $t_i=v_i$, deterministically. Since the signal is a deterministic function of the unknown state of the world, we will again denote  with $\pi\in \Delta(V^n)$ the distribution of the unknown valuation vector, which is the parameter that we wish to identify. Here,  each player first draws a valuation (as an element of the state of the world), and then each player's own valuation is revealed to the player through a signal. After that, a signal is further revealed before players play a BNE given this signal.

In this setting the sharp identified set is again slightly simplified. The result is stated in the next Theorem.
\vskip .1in
\begin{theorem}\label{PV Result}
Let the above private values auction model hold. Given a distribution of bid profiles $\phi$ supported on a set $S\subseteq B^n$, the set of distributions of bids $\pi\in \Delta(V^n)$ that are consistent with $\phi$ are the ones where the following linear program, denoted $LP(\phi, b)$, is feasible: 
\begin{align*}
&~~~ \forall v_i^*\in V, b_i^*,b_i'\in B:~& & \sum_{\vec{v}: v_{i} = v_i^*, \vec{b}\in S: b_i=b_i^*} \phi(\vec{b}) \cdot x(\vec{v}|\vec{b})\cdot \left(u_i(\vec{b};v_i)-u_i(b_i',\vec{b}_{-i};v_i)\right) \geq 0\\
&~~~ \forall \vec{v}\in V^n:~& & \pi(\vec{v}) = \sum_{\vec{b}\in S} \phi(\vec{b})\cdot x(\vec{v}|\vec{b})\\
&~~~ \forall \vec{b}\in S:~& & x(\cdot | \vec{b})\in \Delta(V^n)
\end{align*}

\end{theorem}
Observe, that in this setting, the latter linear program is also linear in $\pi$. Thus we get that the sharp identified set is a convex set and is defined as the set of solutions to the above linear program, where $\pi$ is also a variable. Note also here that no assumption is made on the correlation between player valuation. This result allows for the recovery of valuation distribution with arbitrary correlation (and general signaling structures). As a special case of the above, we study next the IPV model of auctions.  

\paragraph{Independent Private Values.} The situation becomes more complex if we also want to impose an extra assumption that the distribution of private values is independent. In that case, we  have the extra condition that $\pi$ must be a product distribution which is a non-convex constraint. For instance, if we want to assume that the value of each player is independently drawn from the same distribution $\rho$, and hence $\rho$ is what we wish to identify, then we also have the extra constraint that:
\begin{equation}
\forall \vec{v}\in V^n: \pi(\vec{v}) = \rho(v_1)\cdot \ldots\cdot \rho(v_n)
\end{equation}
Adding this constraint into the above LP, makes the LP non-convex with respect to the variables $\rho(v)$ (even though checking whether a given $\rho$ is in the identified set, is still an LP). Thus in this case we cannot compute in polynomial time upper and lower bounds on the moments of the distribution $\rho$ using the above LP.

However, we make the following observation which simplifies the constraints of the LP: we note that conditional on a player's valuation and on a bid profile $b$, the effect of a deviation $u_i(\vec{b};v_i)-u_i(b_i',\vec{b}_{-i};v_i)$ is independent of the values of opponents, in a private value setting. Thus we can re-write the best response constraint as:
\begin{align}
\sum_{\vec{b}\in S: b_i=b_i^*} \phi(\vec{b}) \cdot x_i(v_i^*|\vec{b})\cdot \left(u_i(\vec{b};v_i^*)-u_i(b_i',\vec{b}_{-i};v_i^*)\right) \geq 0
\end{align}
where $x_i(v_i|b) = \Pr[V_i=v_i | B=b]$, where $V_i$ is the random variable representing player $i$'s value and $B$ is the random variable representing the bid profile at a BCE.

Then we can formulate the consistency constraints, by simply imposing a constraint per player, i.e. if $\rho_i$ is the distribution of player $i$'s value, then it must be that:
\begin{equation}
\rho_i(v_i) = \sum_{\vec{b} \in S} \phi(\vec{b})\cdot x_i(v_i|\vec{b})
\end{equation}
These are constraints that are still linear in $\rho_i(v_i)$. We state the LP as a corollary next.

\ \
\begin{corollary}
	Assume that the above IPV model hold. 
The following LP, denoted $LP(\phi, \vec{\rho})$, characterizes the sharp identification under the independent private values model: 
\begin{align*}
&~~~ \forall i\in[n], v_i\in V, b_i^*,b_i'\in B:~& & \sum_{\vec{b}\in S: b_i=b_i^*} \phi(\vec{b}) \cdot x_i(v_i|\vec{b})\cdot \left(u_i(\vec{b};v_i)-u_i(b_i',\vec{b}_{-i};v_i)\right) \geq 0\\
&~~~ \forall i\in [n], v_i\in V:~& & \rho_i(v_i) = \sum_{\vec{b}\in S} \phi(\vec{b})\cdot x_i(v_i|\vec{b})\\
&~~~ \forall \vec{b}\in S:~& & x_i(\cdot | \vec{b})\in \Delta(V)
\end{align*}
\end{corollary}

 \ \
In particular if we assume that player's are symmetric, i.e. $\rho_i(v) = \rho(v)$, then we can compute upper and lower bounds on any moment $E[f(v)]$ of the common value distribution:
\begin{align*}
&~~~  & &\max_{x_i(\cdot | \vec{b})} \sum_{v\in V} f(v)\cdot \rho(v) &\\ 
&~~~\forall i\in[n], v_i\in V, b_i^*,b_i'\in B:~& & \sum_{\vec{b}\in S: b_i=b_i^*} \phi(\vec{b}) \cdot x_i(v_i|\vec{b})\cdot \left(u_i(\vec{b};v_i)-u_i(b_i',\vec{b}_{-i};v_i)\right) \geq 0\\
&~~~ \forall i\in [n], v\in V:~& & \rho(v) = \sum_{\vec{b}\in S} \phi(\vec{b})\cdot x_i(v|\vec{b})\\
&~~~ \forall i\in [n], \vec{b}\in S:~& & x_i(\cdot | \vec{b})\in \Delta(V)
\end{align*}
A by-product of this analysis is that the linear program allows us to test for symmetry in the independent private values model. In particular it is not clear that when we assume that all marginals are the same, then the LP is feasible. Thus by checking feasibility of the LP we can refute the assumption of symmetric independent private values.

\ \ \

\etedit{ Note here that given that we allow the augmenting signals to be arbitrary correlated, we are not able to infer  any information about the joint distribution of valuation, such as correlation, given a bid profile. }
\vscomment{The observation below is wrong. We can still infer something about correlation, due to correlation in bids. What we cannot infer at all is what is the correlation in values conditional on a bid profile. This could still be arbitrary and there constraints only on the marginal distribution of each value conditional on a bid profile.}
\vsdelete{\paragraph{Non-identifiability of Correlation in Values.} The above discussion shows a stronger point: the BCE constraints for the case of private values and under the assumption that players observe their own valuation, yields \emph{no constraints on the correlation of player valuations}, but only constraints the marginal distributions of each player's value. 
Thus assuming that the observed distribution is the outcome of a BCE, does not allow for identification of the correlation among player's valuations. The reasoning behind this result is that it in a Bayes-correlated equilibrium it is impossible to distinguish between the case where valuations are genuinely correlated and where bids are correlated through signaling. Since a BCE allows for arbitrary such signaling, no inference can be made on the underlying correlation of values. We can learn the marginal distributions but the model under general information structures contains no information on the copula without further restrictions on the model.}

\subsection{Point Identification in First Price Auctions with Continuous Bids}

\vscomment{If we make the assumption that the equilibrium is strictly monotone, i.e. that there is a one-to-one mapping between bids and values, then we should be able to extend this to the correlated values case, and claim that marginals are uniquely identified in the correlated values case and the formula which I think is the same as some prior work on correlated first price auctions is robust to informational assumptions. The paper that does exactly that without informational robustness is \cite{Li2002}. So we would be stating something like the results of \cite{Li2002} for estimating private values in the affiliated private values setting is robust to }

Generally in the independent private values model, if we allow the bids and valuations to be continuous, we show here that the bidder specific valuation distributions are all identified even allowing for general information structures. This is important since the information here is allowed to be correlated but yet the valuation distributions $\rho_i$ of each player $i$ (under independence) are point identified. 

Denote with $G_{-i}(\cdot|b_i)$, the CDF of the maximum other bid at a BCE, conditional on a bid $b_i$ of player $i$, and let $g_{-i}(\cdot|b_i)$ be the density. These are observable in the data. Now the utility of a player from submitting a bid $b_i'$ conditional on being recommended by BCE to player $b_i$ and observing a value of $v_i$ is:
\begin{equation}
U_i(b_i'; v_i,b_i) = (v_i - b_i') \cdot G_{-i}(b_i' | b_i)
\end{equation}
Since, $b_i$ is maximizing the above quantity, by the best-response constraints, then we get that the derivative of the utility with respect to $b_i'$ has to be equal to $0$ at $b_i'=b_i$. The latter implies:
\begin{equation}
(v_i-b_i) \cdot g_{-i}(b_i | b_i) - G_{-i}(b_i | b_i) = 0
\end{equation}

We summarize our result in the following Theorem. 

\begin{theorem}
Let the IPV model above hold and further assume the bid distribution is continuous and admits a density.
Thus we can write the value of a player as a function of his bid and of the conditional maximum other bid distribution:
\begin{equation}\label{eqn:inversion}
v_i = b_i +\frac{G_{-i}(b_i | b_i)}{g_{-i}(b_i | b_i)}
\end{equation}
Hence, the valuation distribution for each player is point identified. 
\end{theorem}

Hence, if we know the population bid distribution and if we assume that it is continuous and admits a density, then given the bid of a player we can invert and uniquely identify his value. Thus we can write the CDF of the value distribution of each player as a function of the observables. 

\vsedit{A similar result was shown for the first price IPV model (without signals) in \cite*{guerre2000} and later also generalized to the affiliated private values setting by \cite*{Li2002}. Interestingly, the inversion formula that we arrived to in the independent private values setting, but under robustness to the information structure, is the same as the inversion formula of \cite*{Li2002} under the correlated private value setting. It is interesting to note that under the assumption that the bid of a player is strictly increasing in his value at any Bayes-Correlated equilibrium, then we can re-do the analysis in this section in the more general correlated private values setting and show that we can invert the value of a player from the observed correlated bid distribution. The inversion formula is identical to Equation \eqref{eqn:inversion} and identical to the inversion derived in \cite*{Li2002}. This  essentially shows that in terms of inference, robustness to information structure is in some sense equivalent to robustness to correlation in values. }

Finally, as in the common values case, it is possible to adjust the LP to account for observing only the winning bids in the Private Values case.

%% file: inference.tex
In general, we do no have access to the distribution $\phi(\vec{b})$. Instead, we have access to $N$ i.i.d. observations \vsedit{$\omega_1,...,\omega_N$ from said distribution ; let $\omega_t(\vec{b}) = 1\{\omega_t = \vec{b}\}$. The sampled distribution is then given by 
\begin{align}\label{eqn:finite-sample-distribution}
\phi_N(\vec{b}) = \frac{1}{N} \sum\limits_{t=1}^N \omega_t(\vec{b})
\end{align}} 
In this section, we develop techniques for estimating the sharp identified set using samples from $\phi(\vec{b})$. We showcase our techniques for the CV setup for simplicity. Also, we focus on inference on the identified set for $\pi$ which is the object of inference here. Often times, we parametrize the distribution of the common values, and so inference will be on the identified set for the vector of parameters\footnote{It is possible to construct the CI for the unknown parameters -rather than the identified set by inverting test statistics. }. \etedit{We start with finite sample approaches to constructing confidence regions for sets using concentration inequalities. These seem to be the first application of such results on using such inequalities to settings with partial identification\footnote{Finite sample inference results are particularly attractive in models with partial identification since standard asymptotic approximations are not typically uniformly valid especially in models that are close to/or are point identified.}. We also show how existing set inference methods can also be used to construct confidence regions. }
\subsection{Finite Sample Inference via Concentration Inequalities}\label{sec: nonparam_hoeffding}

We explore first the question of set inference using finite sample concentration inequalities. This allows us to obtain a confidence set for the identified set where the coverage property holds for every sample size. In addition, we highlight tools from the concentration of measure literature applied to inference on sets in partially identified models. 

 As a reminder,  given a probability distribution over bids $\phi(\vec{b})$, the sharp set of compatible equilibria $\phi(v,\vec{b}) = \phi(\vec{b}) x(v|\vec{b})$ is the set $\Theta_I$  of joint probability distributions satisfying: 
\begin{align}\label{eq: limit_constraints}
~~~&~~~ \forall b_i^*,b_i'\in B:~& & \sum_{v\in V,~ \vec{b}\in S: b_i = b_i^*} \phi(\vec{b}) \cdot x(v | \vec{b})\cdot \left(u_i(\vec{b};v)-u_i(b_i',\vec{b}_{-i};v)\right) \geq 0
\end{align} for the proper $x(.|.)$. See the statement of Theorem (\ref{CommonValue Result}) above.
 Let $\pi(.)$ be the vector characterizing the distribution of the common value. Let $\{F_j(x;\pi, \phi): j\in M\}$, denote the negative of the best-response and density constraints, associated with the  BCE LP in Theorem (\ref{CommonValue Result}). Then in the population  $\pi$ is feasible iff:
\begin{equation}
\min_{x}\max_{j\in M} F_j(x;\pi, \phi) \leq 0
\end{equation}
Then we have that the identified set is defined as:
\begin{equation}
\Pi_I = \{\pi: \min_{x} \max_{j\in M} F_j(x;\pi, \phi) \leq 0\}
\end{equation}
Now we consider a finite sample analogue. Observe that $F_j(x;\pi, \phi) = \E[f_j(x;\pi,\omega)]$, where expectation is over the random vector $\omega$ and the function $f_j(\cdot;\pi,\omega)$ takes the form:
\begin{equation}\label{eqn:ex-post-constraint}
f_j(x;\pi,\omega) = \sum_{v\in V,~ \vec{b}\in S: b_i = b_i^*} \omega(\vec{b}) \cdot x(v | \vec{b})\cdot \left(u_i(b_i',\vec{b}_{-i};v)-u_i(\vec{b};v)\right)
\end{equation}
for some triplet $(i, b_i^*, b_i')$ for the case of best-response constraints and similarly for density consistency constraints. Then we consider the sample analogue:
\begin{equation}\label{eqn:finite-sample-constraint}
F_j^N(x;\pi) = \frac{1}{N} \sum_{t=1}^N f_j(x;\pi, \omega_t)
\end{equation}
Observe that due to the linearity of $F_j$ with respect to $\phi$, we can re-write: $F_j^N(x; \pi)=F_j(x;\pi, \phi_N)$. One can then define the estimated identified set by analogy, i.e., replacing $\phi(b)$ with $\phi_N(b):$
\begin{equation}\label{nonparametric set estimate}
\Pi_I^N = \{\pi: \min_{x} \max_{j\in M} F_j(x;\pi, \phi_N) \leq \sigma^N\}
\end{equation}
for some decaying tolerance constant $\sigma^N$ (which can be set to zero).

\vsedit{Because the pdf of a non-parametric distribution is a very high-dimensional object, inference on it will require many samples. Hence, we will instead focus on two more structured inference problems. In the first one we are interested in inferring the identified set for {\it a moment} of the distribution and in the second we make assume that the said pdf is known up to a finite dimensional parameter and infer the identified set of these lower dimensional parameters. One can in principle recover non-parametric inference by simply making the parameters be the values of the pdf at the discrete support points, albeit at a cost in the sample complexity.}

\vsedit{\subsubsection{Inference on Identified Set for Moments}}
We begin with the non-parametric setting and show how to construct inference on the identified set of any moment function $m:V\rightarrow [-H,H]$ of the common value distribution that is valid in finite samples. 
\vsdelete{Even though the identified set on the whole distribution $\pi$ is very high dimensional, we can still answer  probabilistic questions for any moment of the distribution via the Bayesian Bootstrap.}
Observe that the identified set for any moment $m:V\rightarrow [-H,H]$ is an interval $[L,U]$ defined by:
\begin{equation}\label{eqn:moment-identified-set}
\begin{aligned}
L = \min_{x:~ \max_{j\in M} F_j(x;\phi)\leq 0}~& \sum_{v\in V} m(v) \cdot \sum_{b\in S} \phi(\vec{b})\cdot x(v|\vec{b})\\
U = \max_{x:~ \max_{j\in M} F_j(x;\phi)\leq 0}~& \sum_{v\in V} m(v) \cdot \sum_{b\in S} \phi(\vec{b})\cdot x(v|\vec{b})
\end{aligned}
\end{equation} 
where $x$ in both optimization problems is ranging over the convex set of conditional distributions, i.e. $x(\cdot|\vec{b})\in \Delta(V)$, which we omit for simplicity of notation.
We can then define their finite sample analogues as:
\begin{equation}\label{eqn:finite-moment-identified-set}
\begin{aligned}
L^N(\sigma^N) = \min_{x:~  \max_{j\in M} F_j^N(x)\leq \sigma^N}~& \sum_{v\in V} m(v) \cdot \sum_{b\in S} \phi_N(\vec{b})\cdot x(v|\vec{b})\\
U^N(\sigma^N) = \max_{x:~ \max_{j\in M} F_j^N(x)\leq \sigma^N}~& \sum_{v\in V} m(v) \cdot \sum_{b\in S} \phi_N(\vec{b})\cdot x(v|\vec{b})
\end{aligned}
\end{equation} 
Where $F_j^N(x)$ and $\phi_N$ are given in Equations \eqref{eqn:finite-sample-constraint} and \eqref{eqn:finite-sample-distribution} respectively.

The following result gives finite sample high probability bounds on the coverage of the interval $[L^N(\sigma^N)-\epsilon^N, U^N(\sigma^N)+\epsilon^N]$. As a reminder, we use $n$ to designate the number of bidders, $N$ is sample size, and $|B|$ is the number of support points for bids.

\begin{theorem}\label{nonparam inequa bound} 
Suppose that $u_i(\vec{b};v)\in [-H,H]$ for all $i\in [n]$, $\vec{b}\in B^n$ and $v\in V$ and let $\sigma^N =  2H\sqrt{\frac{\log(4n|B|^2/\delta)}{N}}$ and $\varepsilon^N =  2H \sqrt{\frac{\log(4/\delta)}{N}}$. Then:
\begin{align}
\Pr\left[[L,U]\subseteq \left[L^N(\sigma^N)-\epsilon^N, U^N(\sigma^N)+\epsilon^N\right]\right] \geq 1-\delta
\end{align}
where $L$, $U$ are defined by Equation~\eqref{eqn:moment-identified-set} and $L^N(\sigma^N)$, $U^N(\sigma^N)$ are defined by Equation~\eqref{eqn:finite-moment-identified-set}.
\end{theorem}
\begin{proof}
We will show that with probability at least $1-\delta/2$: $
L^N(\sigma^N)-\epsilon^N \leq L$. 
The theorem follows by also showing that $U^N(\sigma^N)+\epsilon^N\geq U$ with probability at least $1-\delta/2$ and then using a union bound of the bad events. The second inequality, follows along identical lines by simply replacing $m(v)$ with $-m(v)$. So it suffices to show the first inequality.

We remind the reader the definitions of the two quantities:
\begin{align}
L = \min_{x:~ \max_{j\in M} F_j(x;\phi)\leq 0}~& \sum_{v\in V} m(v) \cdot \sum_{b\in S} \phi(\vec{b})\cdot x(v|\vec{b})\label{eqn:true-lower}\\
L^N(\sigma^N) = \min_{x:~  \max_{j\in M} F_j^N(x)\leq \sigma^N}~& \sum_{v\in V} m(v) \cdot \sum_{b\in S} \phi_N(\vec{b})\cdot x(v|\vec{b})\label{eqn:sample-lower}
\end{align}
Let $x^*$ be the optimal solution to the population LP of Equation \eqref{eqn:true-lower}. We will argue that w.h.p. for the choice of $\sigma^N$ it remains a solution of the finite sample LP of Equation \eqref{eqn:sample-lower} and the value of the finite sample LP under $x^*$ is at least the value of the population LP plus $\epsilon^N$.

Observe that $F_j^N(x^*)$ is the sum of $N$ i.i.d. random variables bounded in $[-H, H]$ with mean $F_j(x^*;\phi)$. By Hoeffding's inequality  with probability $1-\kappa$:
\begin{align}\label{eq: hoeffding}
F_j^N(x^*) \leq F_j(x^*;\phi) + 2H \sqrt{\frac{\log(1/\kappa)}{N}}\leq  2H \sqrt{\frac{\log(1/\kappa)}{N}}
\end{align}
where the second inequality follows by feasibility of $x^*$ for the LP of Equation \eqref{eqn:true-lower}.
By a union bound over all $j\in M$ and since $|M|=n|B|^2$ we get that with probability at least $1-\kappa n |B|^2$:
\begin{align}
\max_{j\in M} F_j(x^*;\phi) \leq  2H \sqrt{\frac{\log(1/\kappa)}{N}}
\end{align}
For $\kappa=\frac{\delta}{4n|B|^2}$ and $\sigma^N=2H \sqrt{\frac{\log(4n|B|^2/\delta)}{N}}$ we get that with probability $1-\frac{\delta}{4}$: $\max_{j\in M} F_j^N(x^*) \leq \sigma^N$ and thereby $x^*$ is feasible for the finite sample LP.

Finally, the value of the finite sample solution can be written as 
\begin{equation}
Q^N(x) = \frac{1}{N}\sum_{t=1}^N q(x; \omega_t)
\end{equation}
with: 
\begin{equation}
q(x;\omega_t)= \sum_{v\in V} m(v) \sum_{\vec{b}\in S}\omega_t(\vec{b}) \cdot x(v|\vec{b})
\end{equation}
Thus it is also the sum of $N$ i.i.d. random variables bounded in $[-H,H]$,\footnote{Since $m(v)\in [-H,H]$ and $x(\cdot|\vec{b})\in \Delta(V)$.} and with mean $Q(x;\phi)=\sum_{v\in V} m(v) \cdot \sum_{b\in S} \phi(\vec{b})\cdot x(v|\vec{b})$. Hence, by Hoeffding's inequality, with probability at least $1-\delta/4$: 
\begin{align}\label{eq: hoeffding}
Q_j^N(x^*) \leq Q_j(x^*;\phi) + 2H\sqrt{\frac{\log(4/\delta)}{N}} = L +\epsilon^N
\end{align}

By a union bound we get that with probability at least $1-\delta/2$, $x^*$ is feasible for the finite sample LP and achieves value $Q_j^N(x^*)\leq L+\epsilon^N$. In that case, by the definition of $L^N(\sigma^N)$, we get:
$L^N(\sigma^N)\leq Q_j^N(x^*) \leq L+\epsilon^N$ and the theorem follows.
\footnote{We could have used an empirical Bernstein inequality instead of Hoeffding's inequality and and replace the $H$ in the quantities $\sigma^N$ by $\sqrt{\max_j \Var_N(F_j^N(x^*))}$, where $\Var_N$ is the empirical variance, at the expense of adding a lower order term of $O\left(\frac{H}{N}\right)$ (see e.g. \cite{Mauer2009,Peel2010}). Similarly, for $\epsilon^N$. However, the latter requires knowledge of $x^*$ and taking a supremum over $x^*$ in the latter seems to be as conservative as a Hoeffding bound.}
\end{proof}

\subsubsection{Inference on Identified Set for Parametric Distributions.}\label{sec: param_hoeffding}
For the parametric case, we assume that the distribution $\pi(v)$ is parametric of the form $\pi(v,\theta)$ for some finite parameter set $\Theta$. In the parametric setting we need to augment the constraint set $M$, apart from containing the best-response constraints, to contain the parametric form consistency constraints:
\begin{equation}
\forall v\in V: \pi(v,\theta) = \sum_{v\in V} \phi(\vec{b})\cdot x(v|\vec{b})
\end{equation}
These can also be written of the form $F_j(x;\phi)\leq 0$ for some function $F_j(x;\phi)=\E[f_j(x;\omega)]$.\footnote{Simply add one constraint of the form $\pi(v,\theta) - \sum_{v\in V} \phi(\vec{b})\cdot x(v|\vec{b})\leq 0$ and one of the form $\sum_{v\in V} \phi(\vec{b})\cdot x(v|\vec{b}) - \pi(v,\theta)\leq 0$.} We overload notation and let $M$ be this augmented set of constraints. 
Then the parameter of interest is $\theta$ and the identified set for $\theta$ takes the form:
\begin{equation}\label{eqn:parametric-identified-set 1}
\Theta_I = \{\theta\in \Theta: \min_{x} \max_{j\in M} F_j(x;\theta, \phi) \leq 0\}
\end{equation}
and its sample equivalent 
\begin{equation}\label{eqn:finite-parametric-identified-set 1}
\Theta_I^N\left(\sigma^N\right) = \{\theta\in \Theta: \min_{x} \max_{j\in M} F_j(x;\theta, \phi_N) \leq \sigma^N\}
\end{equation}
for some decaying tolerance constant $\sigma^N$ (which can be set to zero).

The next Theorem provides finite sample high probability coverage bounds for $\Theta_I$.
\begin{theorem}\label{thm:parametric-finite-sample}
Consider $\Theta_I$ and $\Theta_I^N\left(\sigma^N\right)$ as defined in Equations \eqref{eqn:parametric-identified-set 1} and \eqref{eqn:finite-parametric-identified-set 1}. Suppose that $u_i(\vec{b};v)\in [-H,H]$ for all $i\in [n]$, $\vec{b}\in B^n$ and $v\in V$. If $\sigma^N =  2H\sqrt{\frac{\log(|\Theta|\cdot(n|B|^2+|V|)/\delta)}{N}}$, then:
\begin{equation}
\Pr\left[ \Theta_I \subseteq \Theta_I^N\left(\sigma^N\right)\right]\geq 1-\delta
\end{equation}
\end{theorem}
\begin{proof}
To show the statement we need to show that with probability $1-\delta$, for all $\theta \in \Theta_I$ it must be that $\theta \in \Theta_I^N(\sigma^N)$. Equivalently, if $\min_{x}\max_{j\in M} F_j(x;\theta, \phi) \leq 0$ then $\min_{x} \max_{j\in M} F_j^N(x;\theta) \leq \sigma^N$. The latter follows along similar lines as in the proof of Theorem \ref{nonparam inequa bound}. Let $\theta\in \Theta_I$ and let $x^*$ be the solution to the population problem. Then by Hoeffding's inequality and a union bound over $M$,\footnote{Now $M$ also contains the density consistency constraints and thereby $|M|=n|B|^2 + |V|$.} for the chosen $\sigma^N$ we have that with probability at least $1-\frac{\delta}{|\Theta|}$: $\max_{j\in M} F_j^N(x^*;\theta) \leq \sigma^N$. Hence, the latter inequality also holds for the minimum over $x$ and thereby $\theta\in \Theta_I^N(\sigma^N)$. By a union bound the latter holds uniformly over $\theta$ with probability at least $1-\delta$.
\end{proof}
\etcomment{Is there a way to say something about how loose the bound in (44)? just heuristically? Either way it is fine...}
\input{sample-variance}

\subsection{Alternative Set Inference Approaches}

Here, we use existing set inference approaches and adapt them to the problem at hand. The first approach exploits the mapping between the reduced form parameter (the distribution of bids) and the identified set via the linear program, and the second uses subsampling approaches to set inference.
\subsubsection{Inference via the Bayesian Bootstrap}
In this section, we describe approaches to inference on the identified set via a Bayesian Bootstrap for both parametric and nonparametric models. We start with the nonparametric case. As a reminder, the identified set in this nonparametric case is 

\begin{equation}
\Pi_I = \{\pi: \min_{x} \max_{j\in M} F_j(x;\pi, \phi) \leq 0\}
\end{equation}
 This is a case of a separable problem where knowing $\mathbf \phi^*= \{\phi^*(b_1), \ldots, \phi^*(b_B)\}$ we can solve for  
the identified set via  the mapping:
\begin{equation}\label{nonparametric set estimate 2}
\Pi(\mathbf \phi) =  \{\pi: \min_{x} \max_{i\in M} F_i(x;\pi, \mathbf \phi^*) \leq \sigma^N\}
\end{equation}
for some decaying tolerance constant $\sigma^N$ (which can be set to zero).
 Inference on parameter sets in separable models is analyzed in (\cite{KlineTamer}) where (posterior) probability statements related to the identified set\footnote{Examples of such statements are: the probability that the identified set belongs to a particular set (e.g, when $\Theta_I$ is an interval, the probability that this interval lies in $[a,b]$) or the probability that a particular vector belongs to the identified set, etc.} can be computed using the mapping between the reduced form parameters, $\phi(\mathbf b)$ in this case, and the structural parameters $\pi.$
Intuitively, for every draw $s$ from the posterior for $\phi$ (this is a multinomial distribution and so obtaining draws from the multinomial is simple using the Bayesian bootstrap\footnote{For example, under a limiting uninformative Dirichlet prior for $\phi(\mathbf b)$, the posterior for this $\phi(\mathbf b)$ approaches the Dirichlet posterior $Dir(n_1, \ldots, n_B)$ where $n_j = \sum_{i=1}^N 1[b_i = b_j]$. Therefore, using results that connect the Gamma and Dirichlet distributions, a draw  from given posterior can be approximated by a weighted Bootstrap with Gamma weights. See also \cite{ChamberlainImbens}.}) we can solve via (\ref{nonparametric set estimate 2}) for a ``copy''  of the implied identified set $\Pi_N^{s}$ by solving the LP above. Using this procedure, we can get a sequence $\{\Pi_N^{s}\}_{s=1}^S$ that we can use to answer probability statements about the identified set $\Pi_I.$ The computational constraint in this nonparametric problem is that the parameter of interest $\pi$ is a vector of probabilities with $|V|$ support points. So, the bigger $|V|$ is the larger the number of parameters the more difficult it is to solve for the identified set via (\ref{nonparametric set estimate}) above. \etedit{Hence, with  the finite support condition on the bids standard (Bayesian) inference on $\phi(\mathbf b)$ can be easily mapped into inference on the set $\Pi_I$. }
\etedit{The exact same bootstrap procedure can be used to provide (posterior) probability statement on identified intervals $[L,U]$ that are defined through the mapping in \eqref{eqn:finite-moment-identified-set} above. This is relevant if one is interested in linear functional of the distribution.  Inference on such objects reduces the computational burden substantively.}

\ \ \

In the parametric case, we assume that the distribution function of $v$ belongs to a parametric class that is known up to a finite dimensional parameter $\theta$ and so the LP in Theorem \ref{CommonValue Result} is modified by setting $\pi(v)=\pi(v,\theta)$. Then the parameter of interest is $\theta$ and the identified set for $\theta$ takes the form:
\begin{equation}\label{eqn:parametric-identified-set}
\Theta_I = \{\theta: \min_{x} \max_{j\in M} F_j(x;\theta, \phi) \leq 0\}
\end{equation}
and so the separable mapping between $\mathbf \phi^*$ and $\theta$ takes the form again of
\begin{equation}\label{eqn:finite-parametric-identified-set 1}
\Theta_I\left(\mathbf \phi^*; \sigma^N\right) = \{\theta: \min_{x} \max_{j\in M} F_j(x;\theta, \mathbf \phi^*) \leq \sigma^N\}
\end{equation}
for some decaying tolerance constant $\sigma^N$ (which can be set to zero).

Here again, we can use the approach in (\cite{KlineTamer}) to answer probability statements on $\Theta_I$. The computational step here is much easier than the nonparametric case as now solving for the identified set for every draw from the posterior for $\phi$ is a lower dimensional problem. The Bayesian bootstrap is used to draw vectors of bid probabilities from the multinomial distribution for bids.  

In addition to this Bayesian bootstrap approach, we also provide methods based on subsampling next.

\subsubsection{Parametric Case via Subsampling}\label{sec: param_quantile_estimation}
\etedit{Subsampling methods can be used to conduct inference on the identified set in a  large sample framework. In particular}, let $\theta$ be the parameter of interest. Again, our problem is isomorphic to one where the objective function $Q(\theta)$ is as follows:
\begin{equation}
Q(\theta) = \min_{x} \max_{j\in M}F_j(x;\theta)
\end{equation}
with a the corresponding sample analogue
\begin{equation}
Q^N(\theta) = \min_{x} \max_{j\in M} F_j^{N}(x;\theta)
\end{equation}

We first re-state the above random variable as a minimax problem over convex and compact sets. Let $p\in \Delta_M$ lie on the simplex on $M$ constraints. Then $Q^N(\theta)$ is equivalently defined as:
\begin{equation}
Q^N(\theta) = \min_{x} \sup_{p} \sum_{j\in M} p_j F_j^N(x;\theta)
\end{equation}
and 
\begin{equation}
Q^*(\theta) = \min_{x} \sup_{p} \sum_{j\in M} p_j F_j(x;\theta)
\end{equation}
\etedit{where $Q^*$ is zero for all feasible $\theta's.$}
In addition, following the approach in \cite*{cht}, we get a confidence region for the identified set $\Theta_I =\{\theta: Q^*(\theta) \leq 0\}$ by studying the asymptotic distribution\footnote{Alternatively, we can define the objective function $\tilde Q^N(\theta) = [Q^N(\theta)]_+$  which is the positive part of $Q^N(\theta)$. The identified set now is the minimizer of  $\tilde Q^N(\theta)$ and then similar approaches to \cite{cht} can be used to construct a CI based on subsampling.} of 
$$ \mathcal C_N= \sup_{\theta \in \Theta_I} \min_{x} \sup_{p} \sum_{j\in M} p_j F^N_j(x;\theta) $$ 
Let
 $\tau_N(1-\alpha)$ be the $(1-\alpha)-$quantile of $\mathcal C$ where $\mathcal C$ is the nondegenerate limit of $\sqrt N \mathcal C_N.$
Then, define the $C_N(1-\alpha)$ as follows
$$ C_N(1-\alpha) = \{ \theta: Q^N(\theta) \leq \tau_N^+(1-\alpha)\}$$ where $\tau_N^+(1-\alpha)= \max(\tau_N(1-\alpha),0).$
Notice here that the event $\Theta_I \subseteq C_N(1-\alpha)$ is equivalent to the event $\mathcal C_N \leq \tau_N^+(1-\alpha)$. The next Theorem states the  result.

\ \

\begin{theorem}Assume that 	 $N$ increases to infinity, and  $\sqrt N(\mathcal C_N)$ converges in distribution to $\mathcal C$, a nondegenerate random variable. Then the set $C_N(1-\alpha)$ defined above has the following coverage property:
$$ \lim_{N\rightarrow \infty} \Pr\left[\Theta_I \subseteq C_N(1-\alpha)\right] = 1-\alpha.$$
\end{theorem}

\ \ \
The asympotic distribution $\mathcal C$ above can be characterized using for example results from \cite{Shapiro2009ch5} and sufficient conditions for nondegeneracy are given there which requires second moments to hold (these hold in our case trivially since we maintain the assumptions that bids take finitely many values). 

Finally, to be able to feasibly implement the above approach, one needs to get a value for the cutoff $\tau_N(1-\alpha)$. It is clear here that the standard nonparametric bootstrap may not work. Even if the asymptotic distribution is normal\footnote{The distribution of $\mathcal C$ is likely to be the supremum of a Gaussian process since the optimal value of the LP is generally not unique.}, we may not be able to estimate its variance because it depends on the number of binding constraints. Given the nondegeneracy of the limit, one approach is to use subsampling to compute the cutoff. This can be accomplished by getting $m$ subsamples of the data such that $m/N \rightarrow 0$ as $N \rightarrow \infty.$ For every subsample, we compute $\mathcal C^m_N$ using a preliminary estimate $\widehat \Theta_I$ and using the sequence $\{\mathcal C^m_N\}_{m=1}^{M}$ to get its upper $(1-\alpha)$ quantile. This will result in an estimate of $C_N(1-\alpha)$. We can then use that as our new $\widehat \Theta_I$ and iterate one or two times. This approach was implemented in the Monte Carlo section below and in the empirical application and  it provided adequate results.

%% file: sample-variance.tex
\paragraph{Sample Variance Based Bounds.} We now show how to improve upon the prior analysis and getting tolerance variables $\sigma^N$ whose leading $1/\sqrt{N}$ term does not depend on $H$, but rather on the sample variance of the constraints. We will modify the finite sample identified set $\Theta_I^N(\sigma^N)$ defined in Equation \eqref{eqn:finite-parametric-identified-set 1}, as follows: instead of using a uniform tolerance $\sigma^N$ for all the constraints, we will define tolerance for each constraint differently based on its sample variance:
\begin{equation}\label{eqn:finite-variance-based-set}
\hat{\Theta}_I^N(\lambda, \sigma^N) = \left\{\theta: \min_x \max_{j\in M} \left(F_j^N(x;\theta)
- \lambda \sqrt{\frac{\Var_N(f_j(x;\theta, \omega))}{N}}\right)\leq \sigma^N\right\}
\end{equation}
Where $\Var_N(X)$ for a random variable $X$, denotes the sample variance: 
$$\Var_N(X) = \frac{1}{N(N-1)}\sum_{1\leq t < t'\leq n} (X_t - X_t')^2.$$ The extra variance modification, which is reminiscent of \emph{sample variance penalization} in empirical risk minimization \citep{Mauer2009} and \emph{optimism} in bandit algorithms \citep{Audibert2009} , will allow us to set $\sigma^N$ to be of order $O(H/N)$ rather than $O(H/\sqrt{N})$.

\begin{theorem}
Consider $\Theta_I$ and $\hat{\Theta}_I^N\left(\lambda, \sigma^N\right)$ as defined in Equations \eqref{eqn:parametric-identified-set} and \eqref{eqn:finite-variance-based-set}. Suppose that $u_i(\vec{b};v)\in [-H,H]$ for all $i\in [n]$, $\vec{b}\in B^n$ and $v\in V$. If $\lambda=\sqrt{2\log(2|\Theta||M|/\delta)}$ and $\sigma^N =  \frac{14 H\log(2|\Theta||M|/\delta)}{3(N-1)}$ (with $|M|=n|B|^2+|V|$), then:
\begin{equation}
\Pr\left[ \Theta_I \subseteq \hat{\Theta}_I^N\left(\lambda, \sigma^N\right)\right]\geq 1-\delta
\end{equation}
\end{theorem}
\begin{proof}
To show the statement we need to show that with probability $1-\delta$, for all $\theta \in \Theta_I$ it must be that $\theta \in \hat{\Theta}_I^{N}(\lambda, \sigma^N)$. Let $\theta\in \Theta_I$ and let $x^*$ be the solution to the population problem, i.e. $x^*\in \arg\min_{x}\max_{i\in M} F_i(x;\theta,\phi)$. 

By the empirical Bernstein bound (Theorem 4 of \cite{Mauer2009}), we have that with probability at least $1-\kappa$:
\begin{equation}
F_j^N(x^*; \theta) \leq F_j(x^*;\theta,\phi) + \sqrt{\frac{2\Var_N(f_j(x^*;\theta,\omega))\log(2/\kappa)}{N}}  + \frac{14 H\log(2/\kappa)}{3(N-1)}
\end{equation}
Since $x^*$ is feasible $F_j(x^*;\theta,\phi)\leq 0$ and thereby with probability $1-\kappa$:
\begin{equation}
F_j^N(x^*; \theta) - \sqrt{\frac{2\Var_N(f_j(x^*;\theta,\omega))\log(2/\kappa)}{N}}  \leq  \frac{14 H\log(2/\kappa)}{3(N-1)}
\end{equation}
Hence, also:
\begin{equation}
F_j^N(x^*; \theta) - \sqrt{\frac{2\Var_N(f_j(x^*;\theta,\omega))\log(2/\kappa)}{N}}  \leq  \frac{14 H\log(2/\kappa)}{3(N-1)}
\end{equation}

By the union bound the latter happens for all $j\in M$ with probability $1-\kappa |M|$. Thus setting $\kappa=\frac{\delta}{|\Theta| |M|}$, we have that with probability $1-\frac{\delta}{|\Theta|}$:
\begin{equation}
\max_{j\in M} F_j^N(x^*; \theta) - \sqrt{\frac{2\Var_N(f_j(x^*;\theta,\omega))\log(2|\Theta||M|/\delta)}{N}}  \leq \frac{14 H\log(2|\Theta||M|/\delta)}{3(N-1)}
\end{equation}
The latter then also holds for the minimum over $x$ of the LHS:
\begin{equation}
\min_{x}\max_{j\in M} F_j^N(x; \theta) - \sqrt{\frac{2\Var_N(f_j(x;\theta,\omega))\log(2|\Theta||M|/\delta)}{N}}  \leq \frac{14 H\log(2|\Theta||M|/\delta)}{3(N-1)}
\end{equation}
Hence, $\theta\in \hat{\Theta}_I^N(\lambda, \sigma^N)$ for the $\lambda$ and $\sigma^N$ stated in the Theorem. The Theorem then follows by a union bound over $\Theta$.
\end{proof}

The latter theorem is asymptotically an improvement over Theorem \ref{thm:parametric-finite-sample}, when the variances of the constraints are not as large as their worst case bound of $4H^2$, which is essentially what is assumed in Theorem \ref{thm:parametric-finite-sample}. However, one drawback of this theorem is that the new estimated set $\hat{\Theta}_I^N(\lambda, \sigma^N)$ requires solving more than linear programs. In particular for every parameter $\theta\in \Theta$ we need to solve an optimization problem of the form:
\begin{equation}
\min_x \max_j \left(\langle \alpha_j, x\rangle - \kappa_j \sqrt{\sum_{b, b'\in S} \gamma_{b,b'} \left(\sum_{v} (x_{vb} \beta_{jvb} - x_{vb'} \beta_{jvb'})\right)^2}\right)
\end{equation}
The negative part is a concave function of $x$, as it can be thought of as a norm of a vector that is a linear function of $x$. Thus this is a non-convex minimization problem. Thus even though it offers a statistical improvement, it is at the expense of computational efficiency.

%% file: monte-carlo.tex
In this section, we examine the techniques of Section~\ref{sec: estimation} to obtain an estimated set on simulated data. We assume in all the simulations that the values and bids are taken from discrete sets, respectively $V$ and $B$. In the whole section, we fix the following parameters: i) The maximum value/bid $H$ is given by $H = 20$.
ii) The set of possible bids is given by $B = \{0, \ldots,H\}$. The set of possible common values is $V = B = \{0, \ldots, H\}$.
We generate equilibrium observations as follows: First, we generate a density $f(v)$ of valuations with support $V$. We consider the four following densities:
\paragraph{Normal density.} $f_n(.)$ is the density of a normal random variable with mean parameters $\mu = 4$ and standard deviation parameter $\sigma = 1$, discretized and truncated to have support $V$, i.e.
\begin{align}
f_n(v) = \frac{\exp(-(v-\mu)^2/\sigma^2)}{\sum\limits_{w \in V} \exp(-(w-\mu)^2/\sigma^2)}
\end{align}
\paragraph{Poisson density.} $f_p(.)$ is the density of a Poisson distribution with parameters $\lambda = 4$, truncated inside $V$, i.e.:
\begin{align}
f_p(v) = \frac{\lambda^v \exp(-\lambda)/v!}{\sum\limits_{w \in V} \lambda^w \exp(-\lambda)/w!}
\end{align}
\paragraph{Binomial density.} $f_b(.)$ is the density of a binomial random variable with probability $p = 0.2$ and number of draws $n = H = 20$, i.e.:
\begin{align}
f_b(v) = {N\choose{v}} p^v (1-p)^{H-v}
\end{align}
\paragraph{Geometric density.} $f_g(.)$ is the density of a geometric random variable with probability $p = 0.2$ truncated to have support $V$, i.e.:
\begin{align}
f_g(v) = \frac{p (1-p)^v}{\sum\limits_{w \in V} p (1-p)^w}
\end{align}
For each of those densities, we then generate one distribution of equilibrium bids $\phi$, through solving the BCE linear program for the given distribution of values and with variables $\phi(\vec{b})$ rather than $\pi(v)$. We then generate $N$ samples of bid vectors from $\phi$, to generate the observed empirical bid distribution $\phi_N$.

In the non-parametric setting, since we cannot directly formulate the estimated set of the variance (as it cannot be written as $E[m(v)]$ for some $m$), we obtain \etedit{a superset} for the identified set  as follows: First obtain upper and lower bounds $E_{min}, E_{max}, E^{2nd}_{min}, E^{2nd}_{max}$ for the first and second moments respectively. Then set the bounds on the standard deviation by the conservative ones: $
E^{2nd}_{min} - E_{max}^2 \leq Var \leq E^{2nd}_{max} - E_{min}^2$. \etedit{A computationally more tedious procedure of estimating the identified set for the variance is to first get the identified set for the (discrete) distribution of $V$ and then using that we can ``solve'' for a bound on the variance\footnote{For example, for every ``draw'' from the identified set for the distribution of $V$, we can obtain a variance. We can repeat the process to build the identified set for the variances.}. }

In the parametric case, bounds on the variance can be obtained directly from recovering bounds on the possible parameters for the distribution we consider, and tight bounds on the parameters imply tight bounds on the second and higher order moments. All simulations in the parametric case are therefore presented in terms of identified and estimated sets on the parameters of the distribution of values. 

\subsection{Parametric vs. Non-Parametric Identified Sets}

In this section, we compare the identified sets when we use and do not use parametric knowledge of the distribution of $v$. Figure~\ref{fig: param_vs_nonparam} uses the true distribution $\phi$ and shows:
\begin{itemize}
\item The set of (mean, standard deviation) pairs that belong to the identified set when solving the linear program with parametric constraints, in brown. 
\item The set of (mean, standard deviation) pairs that belong to the identified set when solving the original LP without parametric constraints, in green.
\end{itemize}
We do so for the four different distributions of the common value (Gaussian, Poisson, binomial and geometric) mentioned above. We remark that the non-parametric linear programs seems to recover the mean of the common value accurately in all cases. However, the bounds obtained on the second moment/standard deviation of the distribution of common values are far from being tight. 

\begin{figure}[h]
\begin{subfigure}{.45\textwidth}
  \centering
  \includegraphics[width=1\linewidth]{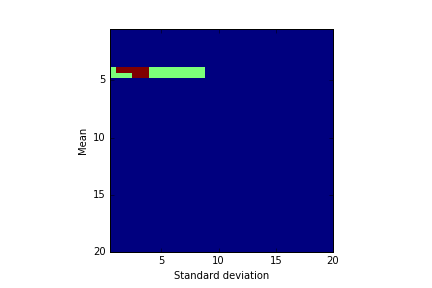}
  \caption{Gaussian density}
\end{subfigure}
\begin{subfigure}{.45\textwidth}
  \centering
  \includegraphics[width=1\linewidth]{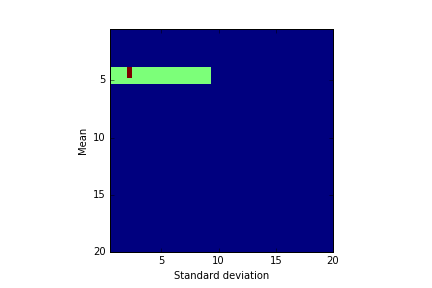}
  \caption{Poisson density}
\end{subfigure}\\
\begin{subfigure}{.45\textwidth}
  \centering
  \includegraphics[width=1\linewidth]{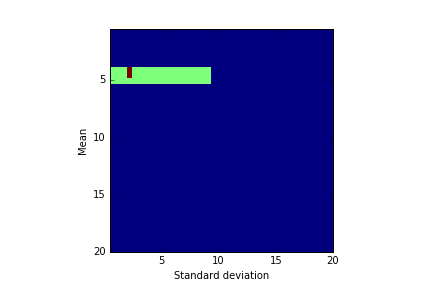}
  \caption{binomial density}
\end{subfigure}
\begin{subfigure}{.45\textwidth}
  \centering
  \includegraphics[width=1\linewidth]{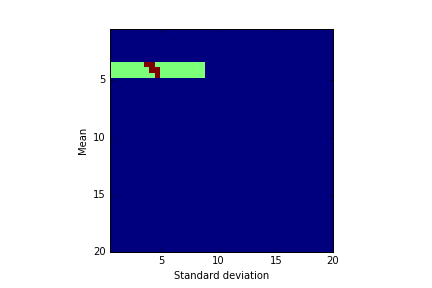}
  \caption{geometric density}
\end{subfigure}
\caption{Plots of the non-parametric and parametric identified sets in the (mean, standard deviation) space}
\label{fig: param_vs_nonparam}
\end{figure}

Figure~\ref{fig: nonparam_hoeffding} compares the identified set for the true distribution to the estimated set using Hoeffding with $\delta = 0.10$, as described in Section~\ref{sec: nonparam_hoeffding} for the Gaussian density function, for a number of samples $N \in \{10^3, 10^4, 10^5, 10^6\}$. We see that as $N$ grows larger, the estimated set grows smaller and smaller and closer to the true identified set. 

\begin{figure}[h]
\begin{subfigure}{.45\textwidth}
  \centering
  \includegraphics[width=1\linewidth]{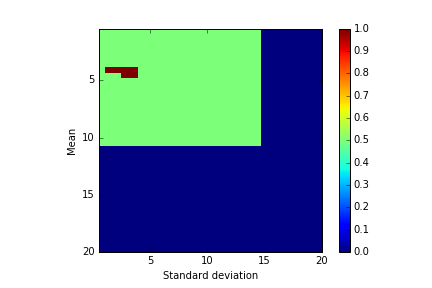}
  \caption{N = 1000}
\end{subfigure}
\begin{subfigure}{.45\textwidth}
  \centering
  \includegraphics[width=1\linewidth]{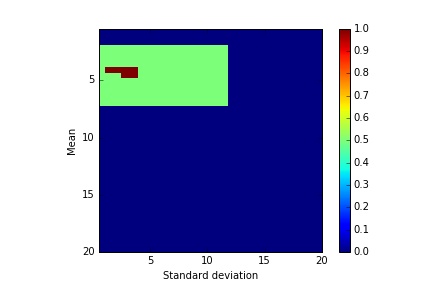}
  \caption{N = 10000}
\end{subfigure}\\
\begin{subfigure}{.45\textwidth}
  \centering
  \includegraphics[width=1\linewidth]{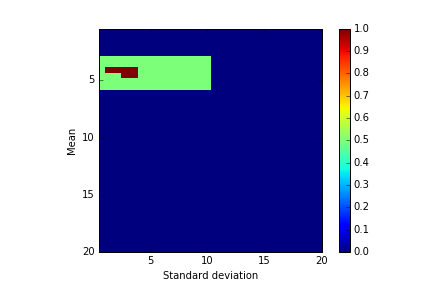}
  \caption{N = 100000}
\end{subfigure}
\begin{subfigure}{.45\textwidth}
  \centering
  \includegraphics[width=1\linewidth]{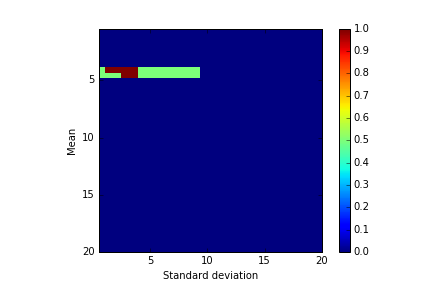}
  \caption{N = 1000000}
\end{subfigure}
\caption{Plots of the parametric identified sets and the non-parametric estimated set through Hoeffding with $\delta = 0.10$ in the (mean, standard deviation) space}
\label{fig: nonparam_hoeffding}
\end{figure}

\subsection{Parametric Identified Sets}

In this section, we characterize the identified and estimated sets in the parametric case for the four distributions described above. For the Gaussian distribution, we provide figures of the identified parameters in the $(\mu,\sigma)$ space. For the other, $1$-dimensional parameter distributions, we provide intervals for the parameter of the chosen distribution. We consider two different techniques to determine the estimated set:
\begin{itemize}
\item Using tolerances determined by Hoeffding's inequality, given in Section~\ref{sec: param_hoeffding}
\item Using quantiles of the tolerance via subsampling, as discussed in Section~\ref{sec: param_quantile_estimation}
\end{itemize}
In all figures, the brown region is the true identified set while the union of the brown and the green region is the estimated set. Figure~\ref{fig: param_hoeffding_gaussian} plots the identified and estimated set when using a Gaussian distribution for the common value and tolerances determined through Hoeffding with $90$ percent confidence ($\delta = 0.10$), for the number of samples $N \in \{10^3, 10^4, 10^5, 10^6\}$. We remark that the number of samples needs be large (of the order of at least $10^5$) for Hoeffding to perform well.
\begin{figure}[h]
\begin{subfigure}{.24\textwidth}
  \centering
  \includegraphics[width=1\linewidth]{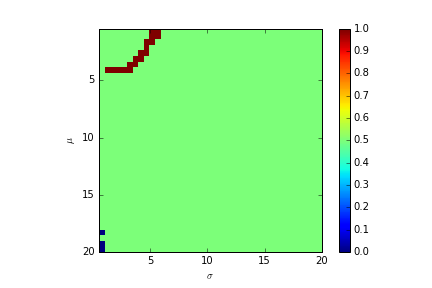}
  \caption{N = 1000}
\end{subfigure}
\begin{subfigure}{.24\textwidth}
  \centering
 \includegraphics[width=1\linewidth]{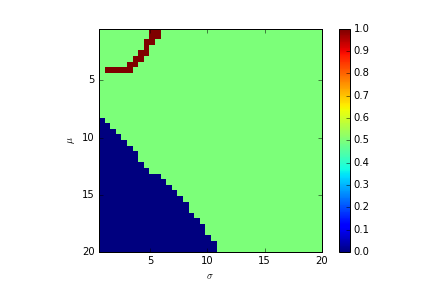}
  \caption{N = 10000}
\end{subfigure}
\begin{subfigure}{.24\textwidth}
  \centering
\includegraphics[width=1\linewidth]{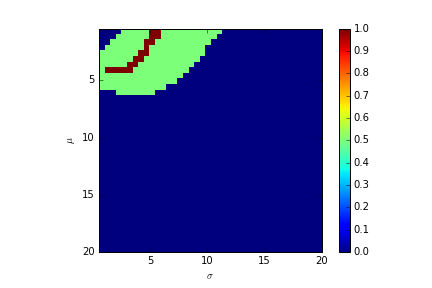}
  \caption{N = 100000}
\end{subfigure}
\begin{subfigure}{.24\textwidth}
  \centering
   \includegraphics[width=1\linewidth]{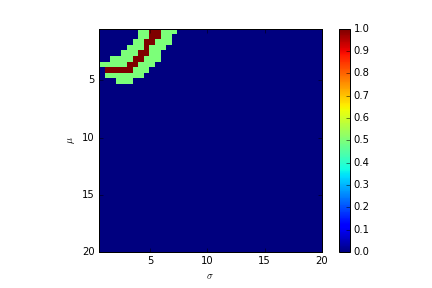}
  \caption{N = 1000000}
\end{subfigure}
\caption{Plots of the parametric identified sets and the parametric estimated set through Hoeffding with $\delta = 0.10$  in the $(\mu, \sigma)$ space}
\label{fig: param_hoeffding_gaussian}
\end{figure}

Figure~\ref{fig: param_quantile_gaussian} plots the identified and estimated set when using a Gaussian distribution for the common value and tolerances determined through subsampling and quantile estimation for the $90$, $95$ and $99$ percent quantiles, as seen in Section~\ref{sec: param_quantile_estimation}; we only use $N = 100$ samples for the bid distribution in the three figures and $k = 50$ subsamples of size $s = N/4 = 25$. We note that the quantile estimation technique covers the true identified set fairly sharply even though $N$ is only equal to $100$ and hence should be preferred to Hoeffding when small amounts of data are available.
\begin{figure}[h]
\begin{subfigure}{.31\textwidth}
  \centering
  \includegraphics[width=1\linewidth]{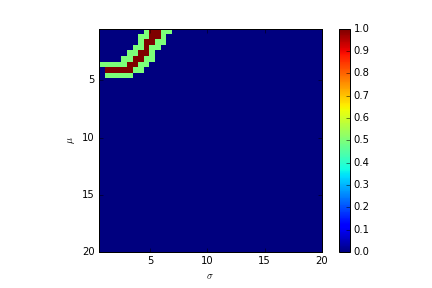}
  \caption{0.90 quantile}
\end{subfigure}
\begin{subfigure}{.31\textwidth}
  \centering
  \includegraphics[width=1\linewidth]{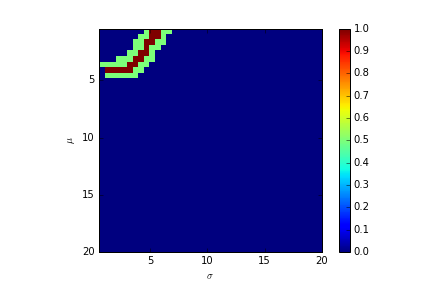}
  \caption{0.95 quantile}
\end{subfigure}
\begin{subfigure}{.31\textwidth}
  \centering
\includegraphics[width=1\linewidth]{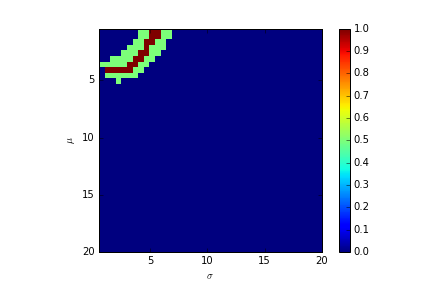}
  \caption{0.99 quantile}
\end{subfigure}

\caption{Plots of the parametric identified sets and the parametric estimated set through the quantile method in the $(\mu, \sigma)$ space for $N=100$}
\label{fig: param_quantile_gaussian}
\end{figure}

Figure~\ref{tab: quantile_other_dists_100} gives the true identified interval for the parameters of the Poisson, binomial and geometric distributions and the estimated interval using subsampling for quantile estimation, for the $90$, $95$ and $99$ percent quantiles -- see~\ref{sec: param_quantile_estimation} -- for a bid distribution sampled with $N = 100$. We see that for the binomial distribution, the estimated set is very close to the true identified set. While the recovered sets for the Poisson distribution are not as sharp as the recovered sets for the binomial distribution, they still restrict the space of possible parameters in a reasonable way: the recovered interval is $[1.5,6.5]$ while the space of possible parameters is $[0,20]$. However, the recovered sets for the geometric distribution contains more than half of the possible parameters, which is unsatisfying. A reason for this comes from the fact that the geometric distribution has a much higher variance than a Poisson or binomial with comparable means, hence there is a lot of variability across different subsamples that can lead to large tolerances. 

\begin{figure}
\begin{footnotesize}
\begin{center}
  \begin{tabular}{ | l | c | c | c | c | }
    \hline
Distribution & True identified set & 0.90 quantile set & 0.95 quantile set & 0.99 quantile set\\ \hline
    Binomial & [0.19,0.22] & [0.17,0.25] & [0.16,0.25] & [0.15,0.27] \\ \hline
    Poisson & [4.0,4.5] & [1.5,6.5] & [1.5,6.5] & [1.5,6.5] \\ \hline
    Geometric & [0.17,0.22] & [0.06,0.56] & [0.04,0.66] & [0.02,0.70] \\
    \hline
  \end{tabular}
\end{center}
\end{footnotesize}
\caption{Identified and estimated sets of the parameters of the distributions (N=100).}\label{tab: quantile_other_dists_100}
\end{figure}

Using $N = 500$ samples and $k = 50$ subsamples of size $s = N/4 = 125$, we obtain Figure~\ref{tab: quantile_other_dists_500}. We can see that the estimated sets for the binomial and geometric distributions are now significantly tighter. 

\begin{figure}
\begin{footnotesize}
\begin{center}
  \begin{tabular}{ | l | c | c | c | c | }
    \hline
Distribution & True identified set & 0.90 quantile set & 0.95 quantile set & 0.99 quantile set\\ \hline
    Binomial & [0.19,0.22] & [0.18,0.24] & [0.17,0.25] & [0.17,0.25] \\ \hline
    Poisson & [4.0,4.5] & [3.5,5.0] & [3.0,5.5] & [2.5,5.5] \\ \hline
    Geometric & [0.17,0.22] & [0.16,0.25] & [0.15,0.28] & [0.14,0.29] \\
    \hline
  \end{tabular}
\end{center}
\end{footnotesize}
\caption{Identified and estimated sets of the parameters of the distributions (N=500).}\label{tab: quantile_other_dists_500}
\end{figure}

Finally, Figure~\ref{fig: param_quantile_gaussian_500} plots the identified and estimated set when using a Gaussian distribution for the common value and tolerances determined through subsampling and quantile estimation for the $90$, $95$ and $99$ percent quantiles, as Figure~\ref{fig: param_quantile_gaussian_500}. However, we now use $N = 500$ samples for the bid distribution in the three figures and $k = 50$ subsamples of size $s = N/4 = 125$. We note that the estimated set now almost coincides with the true identified set. 

\begin{figure}[h]
\begin{subfigure}{.31\textwidth}
  \centering
  \includegraphics[width=1\linewidth]{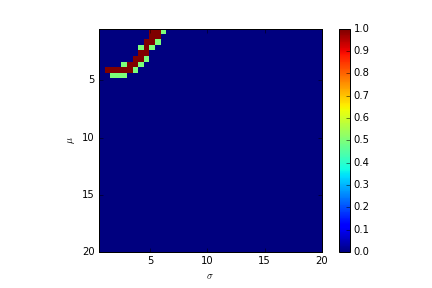}
  \caption{0.90 quantile}
\end{subfigure}
\begin{subfigure}{.31\textwidth}
  \centering
  \includegraphics[width=1\linewidth]{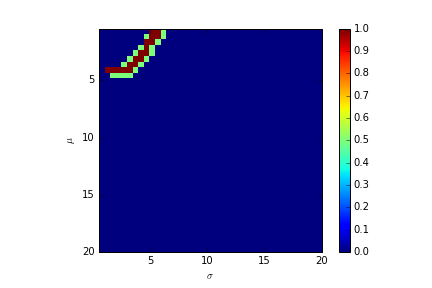}
  \caption{0.95 quantile}
\end{subfigure}
\begin{subfigure}{.31\textwidth}
  \centering
\includegraphics[width=1\linewidth]{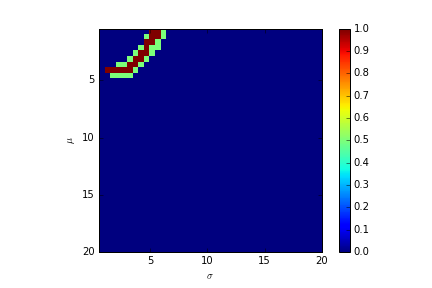}
  \caption{0.99 quantile}
\end{subfigure}
\caption{Plots of the parametric identified sets and the parametric estimated set through the quantile method in the $(\mu, \sigma)$ space for $N=500$}
\label{fig: param_quantile_gaussian_500}
\end{figure}

%% file: ocs-data-analysis.tex
In this section, we illustrate our framework for common value auctions on real data. We use the Outer Continental Shelf (OCS) Auction Dataset that was used in the seminar work of Hendricks and Porter (see \cite{hendricksporter}). The dataset contains bidding information on $3036$ tract auctions in Louisiana and Texas. In particular, for each auction, our dataset contains the acreage of the tract and the total bid of each participant in the auction. We assume that the bidders participate in a first price common value auction, where the value is defined per acre; our goal is i) to show that indeed, the bidders' behavior in the data can be explained by a common value auction (via the testable restriction of whether the estimated identified set is empty), and ii) to recover the first and second order moments of the distribution of said value. This is under weak assumptions on information in that the framework allows bidders in different auctions to know more information about the environment.

\paragraph{Pre-processing of data:}
The dataset contains $3036$ auction with varying number of players. We consider $2$-player common value auctions, hence we only keep the entries in the dataset that contain exactly $2$ bidders; there are $584$ such auctions. We model the two bidders as being the same over the $584$ auctions, and assign bidders' identities to be $1$ or $2$ uniformly at random in each auction. Many of the bids we have are zero and Figure~\ref{fig: bid_distribution} plots the distribution of bids; it has mean $\$991.48$ and standard deviation $\$1825.43$.

\begin{figure}[H]
  \centering \includegraphics[width=.4\linewidth]{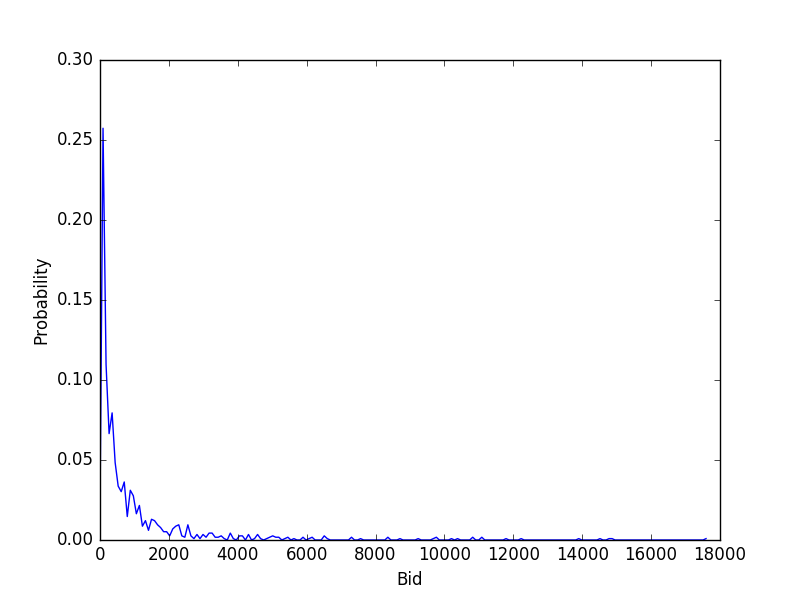}
\caption{Distribution of bids}
\label{fig: bid_distribution}
\end{figure}

 We assume the distribution of the common value per acre has bounded finite support $V = \{0, \ldots, H\}$. We renormalize the bids per acre to be in $[0,\lceil \frac{H}{2} \rceil]$ -- we pick $H/2$ because the common value could have a distribution whose support goes beyond the observed bids --, and discretize the set of bids to be $\{0, \ldots, \lceil \frac{H}{2} \rceil \}$; we do so by rounding each renormalized bid in each auction to the closest integer. We remark that the dataset contains a few outliers whose bid per acre is significantly higher than in all other auctions; we therefore delete the auctions that contain bids over threshold $t = \$20000$. We further assume that the distribution of the common value is given by a truncated normal distribution that takes discrete values in $\{0,\ldots,H\}$, exactly as described in the simulations of Section~\ref{sec: simulations}, and parametrize all optimization problems we solve accordingly. 

\paragraph{Results:}

In all figures, the value of $(\mu,\sigma)$ are given as the values in dollars instead of the corresponding discretized and renormalized value, for the sake of comparison with the $\$20,000$ threshold and the corresponding maximum value of $\$40,000$. Figure~\ref{fig: ocs_heatmap} plots a heat map of the estimated set as a function of the chosen tolerance in the $(\mu,\sigma)$ space for two different values of $H$. Each color on the heat map corresponds to a tolerance level, and the mapping from tolerance levels to colors is given by the colorbar on the right of each figure. The color that is assigned to any given $(\mu,\sigma)$ pair corresponds to the minimum level of tolerance that the analyst needs to add to the equilibrium constraints for $(\mu,\sigma)$ to belong to the estimated set; therefore, the heat map shows how the estimated set grows as the tolerance picked by the analyst increases. 
\begin{figure}[h]
\centering
\begin{subfigure}{.45\textwidth}
  \centering
 \includegraphics[width=1\linewidth]{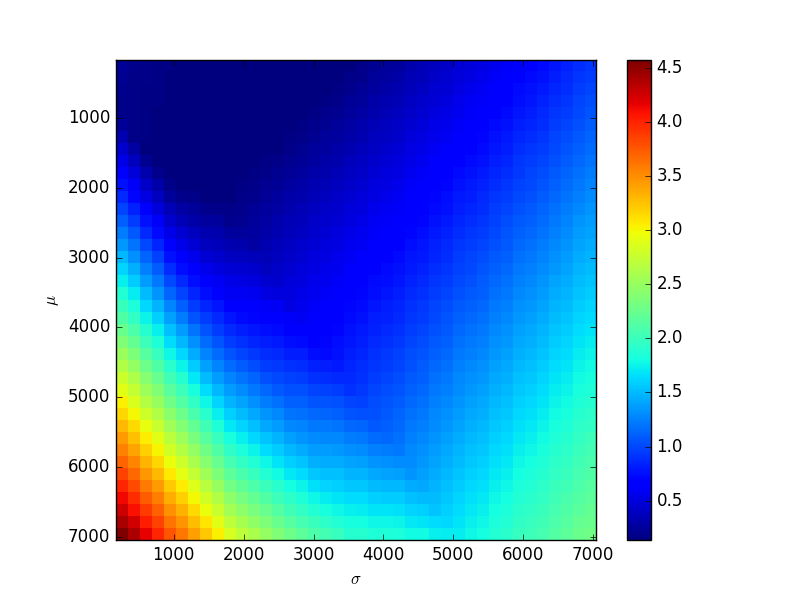}
  \caption{$H = 200$}
\end{subfigure}
\begin{subfigure}{.45\textwidth}
  \centering
 \includegraphics[width=1\linewidth]
{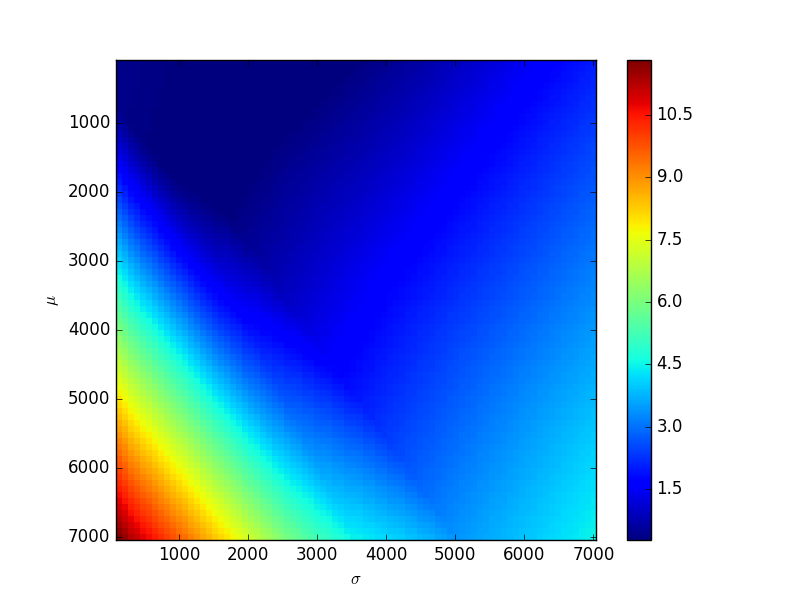}
  \caption{$H = 400$}
\end{subfigure}
\caption{Heat map of minimum tolerance needed for $(\mu, \sigma)$ pairs to belong to the estimated set}
\label{fig: ocs_heatmap}
\end{figure}

Figure~\ref{fig: ocs_qtl_400} plots the estimated set when using the tolerance determined by the subsampling approach of Section~\ref{sec: param_quantile_estimation}, in green using 95\% level, and compares it to the estimated set using the minimum tolerance for which the estimated set is non-empty, in brown, for $H = 400$. No matter what method is used for picking the tolerance and determining the corresponding confidence intervals, the region that is obtained cannot be possibly smaller than the minimum tolerance set unless it is empty; in this sense, the minimum tolerance set is the best possible set that one could hope to obtain through \textit{any} method for determining tolerances. We remark that i) the estimated set for $(\mu,\sigma)$ remains small relatively to the upper bound of $\$20,000$ on the bids and of $\$40,000$ on the maximum common value, and ii) the estimated set is not much bigger compared to the best estimated set we could hope to obtain, indicating that on top of covering the true identified set with high probability, our techniques cannot possibly overestimate the size of the identified by too much. We can see from the confidence regions  that the mean of the common values varies from zero 0 to around \$4000 while the standard deviation varies from close to zero to 6000. It is possible to estimate these means as functions of covariates and hence to allow for observed heterogeneity.

\begin{figure}[h]
\centering
\begin{subfigure}{.31\textwidth}
  \centering
 \includegraphics[width=1\linewidth]{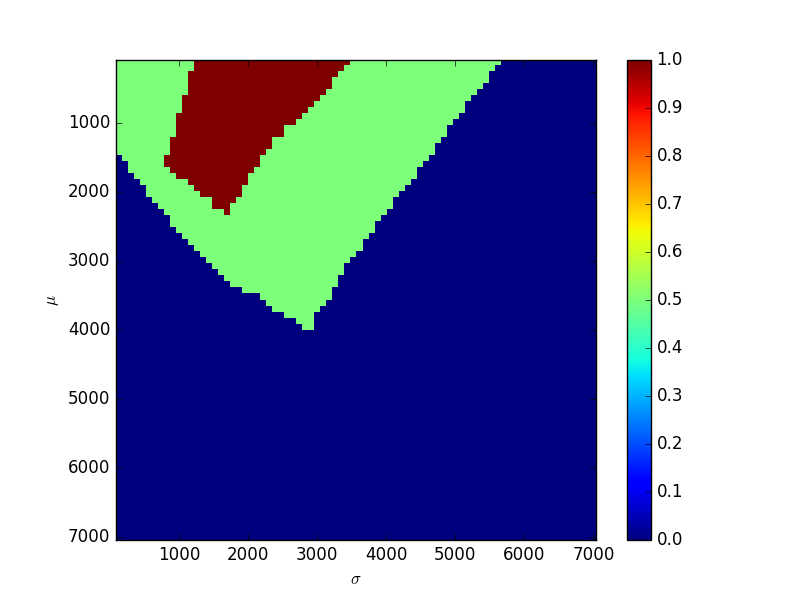}
  \caption{0.90 quantile}
\end{subfigure}
\begin{subfigure}{.31\textwidth}
  \centering
\includegraphics[width=1\linewidth]
{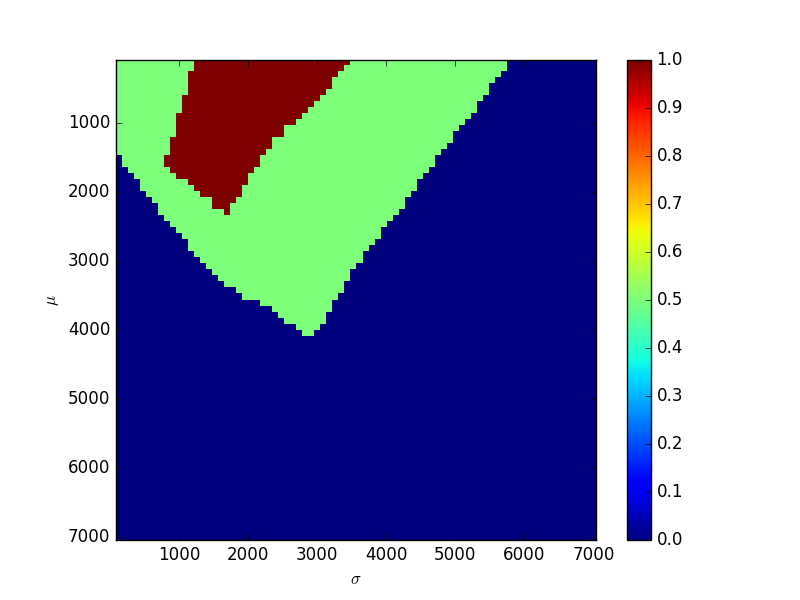}
  \caption{0.95 quantile}
\end{subfigure}
\begin{subfigure}{.31\textwidth}
  \centering
 \includegraphics[width=1\linewidth]
{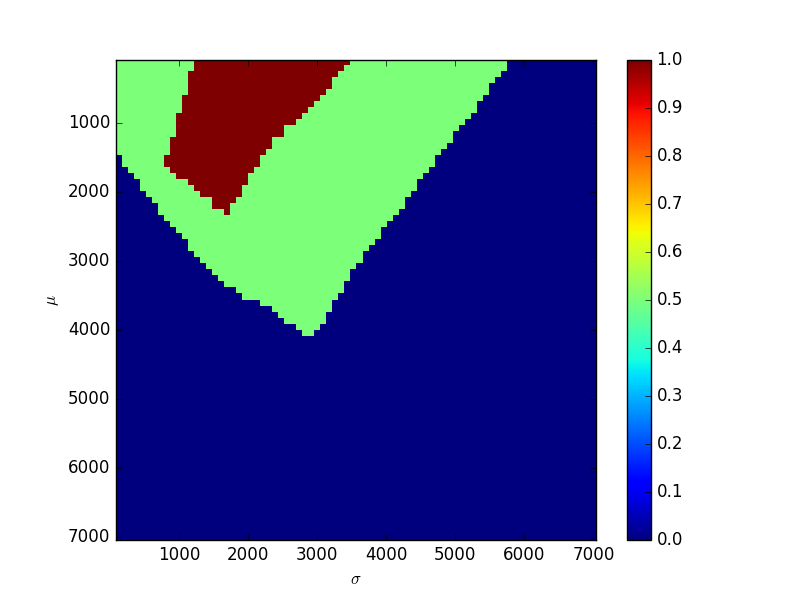}
  \caption{0.99 quantile}
\end{subfigure}
\caption{Estimated sets using the $90$, $95$ and $99$ percent quantiles, for $H = 400$.}
\label{fig: ocs_qtl_400}
\end{figure}

The above plots the mean and variance of a normal density that is then truncated. Since the truncated normal has a different mean and variance then the underlying normal, we plot in
Figure~\ref{fig: ocs_qtl_400_meanstdspace}  the minimum tolerance and the estimated sets as a function of this overall mean and standard deviation of the  distribution of common values that we identify, instead of the parameters $\mu, \sigma$ of the truncated normal distribution. We obtain these plots by computing a mapping from $(\mu,\sigma)$ pairs to (mean, standard deviation) pairs, and plot the image of Figure~\ref{fig: ocs_qtl_400} by said mapping in the (mean, standard deviation) space. We note that the estimated set is fairly small, indicating that our approach identifies the first two moments of the true, underlying distribution of the common value in an accurate fashion, despite only having access to a limited number of samples. 

\begin{figure}[h]
\centering
\begin{subfigure}{.31\textwidth}
  \centering
 \includegraphics[width=1\linewidth]{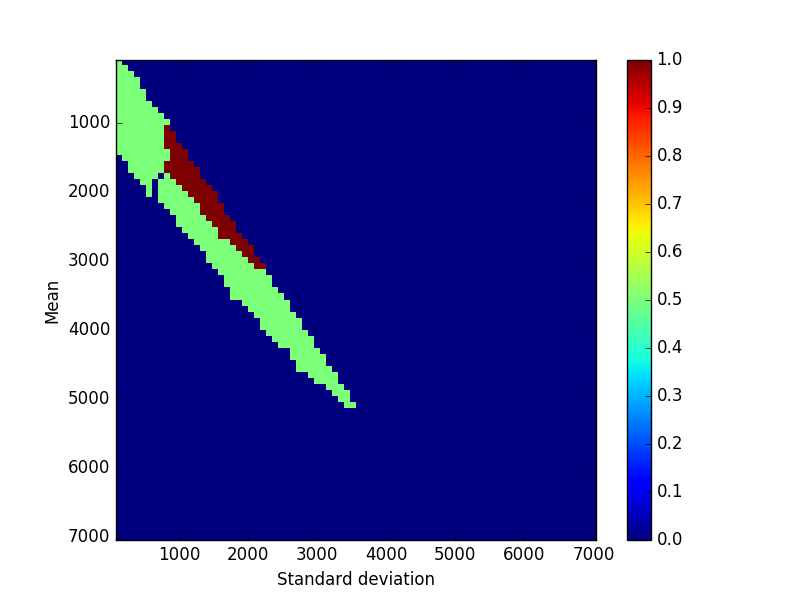}
  \caption{0.90 quantile}
\end{subfigure}
\begin{subfigure}{.31\textwidth}
  \centering
\includegraphics[width=1\linewidth]
{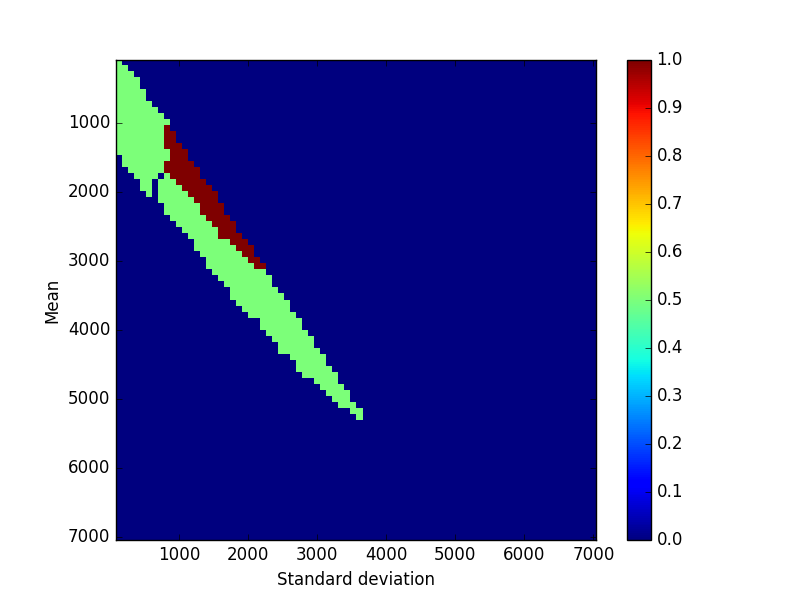}
  \caption{0.95 quantile}
\end{subfigure}
\begin{subfigure}{.31\textwidth}
  \centering
 \includegraphics[width=1\linewidth]
{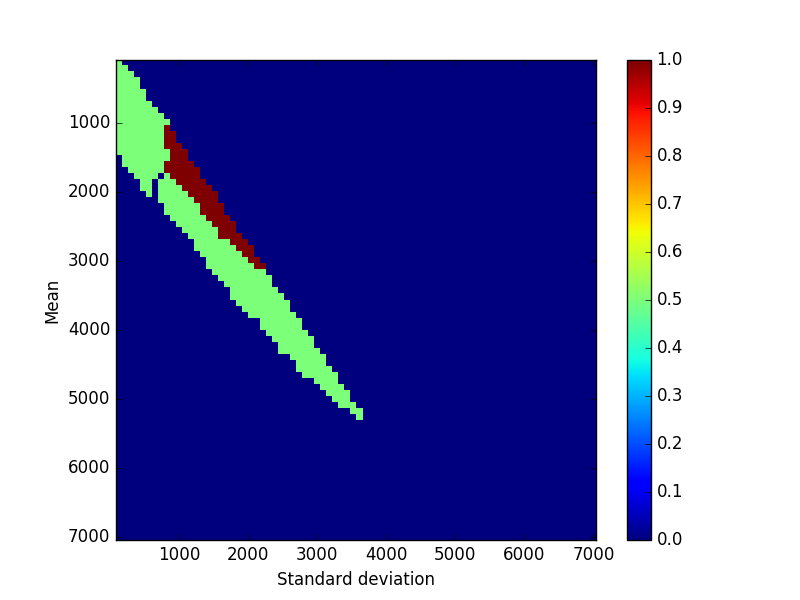}
  \caption{0.99 quantile}
\end{subfigure}
\caption{Estimated sets using the $90$, $95$ and $99$ percent quantiles, for $H = 400$.}
\label{fig: ocs_qtl_400_meanstdspace}
\end{figure}

%% file: appendix.tex
\section{Non-sharpness of the BBM bound}\label{appendix A}
We examine the upper and lower bounds on the mean from BMM. In a common value auction an easy worst-case bound on the mean of the common value distribution comes from the following fact: at any BCE players are getting non-zero expected utility. Thus an obvious lower bound on mean value is the observed revenue. We can try to show that mean value is also not much more than the revenue or some function of the observed bids that is close to the revenue. For instance, the BBM bound shows that:
\begin{equation}\label{eq:brooks-bound}
v(q) - \frac{1}{q^{(n-1)/n}}\int_0^{v(q)} (F(x))^{(n-1)/n} dx \leq B(q)
\end{equation}
where $q\in[0,1]$ is a quantile, $v(q)=F^{-1}(q)$, $F(x)$ is the CDF of the value distribution and $B(q)$ is the quantile function of the maximum bid distribution. This is equivalent to:
\begin{align}\label{eq:brooks-bound-2}
\frac{1}{q^{\alpha}} \int_{0}^{q} \alpha y^{\alpha-1} v(y)dy\leq B(q)
\end{align}
 From this we can try to upper bound the mean of $F$ (i.e. $\int_0^1 v(q) dq$) as a function of the revenue, which is simply the mean of the maximum value distribution (i.e. $\int_0^1 B(q)dq$). If the CDF is convex, $F(0)=0$ and $F(H)=1$, then $F(x)\leq x/H$ hence the BBM bound gives:
\begin{equation}
v(q)\left(1 - \frac{n}{2n-1}\frac{1}{H^{(n-1)/n}}\left(\frac{v(q)}{q}\right)^{(n-1)/n}\right)\leq B(q)
\end{equation}
By integrating and if we let $\mu$ be the mean of the common value and $R$ be the observed revenue of the auction:
\begin{align*}
\mu \leq~& R + \frac{n}{2n-1}\frac{1}{H^{(n-1)/n}}\int_{0}^1\frac{v(q)^{(2n-1)/n}}{q^{(n-1)/n}}dq
\end{align*}
\etdelete{Of course this is far from being a final analysis. But} This shows that there is a reasonable upper bound. For instance for the uniform distribution, i.e. $v(q)=q\cdot H$, this gives $\mu \leq R + \frac{n}{2n-1}H\int_0^1 q dq = R+ H\frac{n}{2n-1}\frac{1}{2}\leq R+H/3$. \etdelete{Not perfect, but non-trivial.} Observe that if the uniform distribution is allowed in our set of possible value distributions then we cannot hope to prove any better upper bound.

For the case of two bidders we can now show that the BBM bound implies that the mean of the auction must be at most:
\begin{equation}
\mu \leq R \cdot \left(2\sqrt{\frac{H}{R}}-1\right)
\end{equation}
where $R$ is the observed revenue of the auction and $H$ is the assumed upper bound on the distribution. Thus if the revenue is at least $H/\alpha$ for some $\alpha$, then we can conclude that the mean must leave in the range: $[R, R(2\sqrt{\alpha}-1)]$ which is a non-trivial sharp identified set. In other words the ratio of the upper to lower bound of the sharp identified set is at most $2\sqrt{\alpha}-1$. 

For $n$ bidders we show (see Section \ref{sec:upper-mean} below) \etdelete{after a lot of analysis} that the BBM bound implies that the mean of the distribution cannot be more than: $\sqrt{\frac{2n}{n-1}\cdot R \cdot H}$ or alternatively $\sqrt{\frac{2n}{n-1}\cdot R \cdot \left(\frac{1}{n}R+L\right)}$ if the density of values is bounded from below by $1/L$. Thus if the revenue is at least $H/\alpha$, then we have that $\mu$ lies in the range: $[R, R\cdot \sqrt{\frac{2n}{n-1}\alpha}]$. In fact the above bound is tight if one only uses the deviation of BBM, i.e. if all we know is the inequality implied by these deviations, then there exists a distribution that satisfies that inequality and which achieves the above bound.

\ \

 {\bf Non-sharpness of BBM:} The upper bound of BBM \etdelete{, even though it gives some non-trivial identified set for the mean, we show here that it } can be very far from the identified set. Before we give a concrete example, we first argue why the bound from \cite{Bergemann2016c} cannot possibly be sharp: The main point of the work of BBM was getting worst-case bounds on the revenue of a common value first price auction as a function of the distribution of values. For that reason in the analysis, the bid distribution is not part of the input. Their main approach  is to claim that the revenue cannot be too small. The idea behind the analysis is as follows: since something is a BCE, no player wants to deviate to any other action. Hence, they do not want to deviate to a specific style of a deviation which is: conditional on your bid deviate uniformly at random above your bid (upwards deviation). Such deviation arguments were also used in \cite*{Hoy2015, Syrgkanis2013} to give bounds on the welfare of non-truthful auctions such as the first price auction.

However, this means that the final bound given in \eqref{eq:brooks-bound} is the product of a subset of the best-response deviation constraints. If we actually knew what the bid distribution is in a BCE (which is the case in the econometrics task), then we wouldn't look at only these bid-oblivious upwards deviations. We would instead compute an optimal best-response bid for this given bid distribution and right the constraint that the player does not want to deviate to this tailored best-response. 

We now give a concrete ``extreme'' example where the upper bound provided by \eqref{eq:brooks-bound} can be as large as twice the sharp bound. In fact, in this example the mean can be point identified, but the bound \eqref{eq:brooks-bound} gives a large interval. 

\ \

\begin{example} Consider the case of two bidders, where the observed maximum bid distribution is a singleton i.e. $\{b^*\}$. First we argue that in this case, we can conclude that the mean $\mu$ of the common value is equal to $b^*$. First, we know that the total expected utility of the bidders must be non-negative, hence the mean of the common value is at least the expected revenue of the auction, which is equal to $b^*$, i.e. $\mu\geq b^*$. 

Now suppose that the mean was strictly larger than $b^*$, i.e. $b^*+\epsilon$. We will show that this yields a contradiction. 

First consider the case where in the support of the BCE, there exist bid vectors of the form $(b,5)$ or $(5,b)$ for $b<5$. Then a player when seeing a bid of $b<5$, he wants to deviate to bidding $b^*+\zeta$ for some $\zeta\in (0,\epsilon)$. The reason is that with a bid of $b$ he knows he is not winning and hence he is getting zero utility, while with a bid of $b^*+\zeta$ he knows he is deterministically winning, getting a value of $\mu$ and paying $b^*+\zeta<\mu$, yielding strictly positive utility. Thus it must be that the BCE contains only one bid vector, i.e. $(b^*,b^*)$. But under this BCE, it is obvious that the expected value conditional on winning is $\mu$ and it is also clear that one of the two players is winning with probability less than $1$, i.e $1-\delta$. Thus if this player deviates to $b^*+\zeta$, he wins deterministically and pays $b^*+\zeta$. Thus the net effect of this deviation is $\mu-b^*-\zeta-(1-\delta)(\mu-b^*)=\delta \cdot (\mu - C) - \zeta=\delta \epsilon - \zeta$.  Taking $\zeta\rightarrow 0$, yields a deviation with positive net effect. 

Thus we get a contradiction, and we conclude that the mean is point identified and  $\mu=b^*$.

Now consider the bound that is derived from Equation \eqref{eq:brooks-bound} for this case of a bid distribution. First we note that for a singleton bid distribution we have that $B(q)=b^*$ for all $q\in [0,1]$. Thus the constraint on the distribution of values implied by Equation \eqref{eq:brooks-bound-2} is simplified to:
\begin{equation}
\forall q\in [0,1]: \frac{1}{\sqrt{q}} \int_{0}^{q} \frac{v(y)}{2\sqrt{y}}dy \leq b^*
\end{equation}
We now give a distribution of values, i.e. a function $v(\cdot)$, that satisfies the above constraint and whose mean $\mu^*$ is at least $2\cdot b^*-(b^*)^2/H$, where $H$ is an externally assumed upper bound on the distribution of values. If $H$ is large, then this gives an interval such that the ratio of the upper to the lower bound is of size $2$, which is far from point identification.

Consider a value function $v(\cdot)$, as follows (two point-mass distribution):
\begin{equation}
v(q) = \begin{cases}
H & \text{ if $q\geq \theta$ }\\
0 & \text{ o.w. }
\end{cases}
\end{equation}
Then the constraint simplifies to:
\begin{align*}
b^* \geq~& \frac{1}{\sqrt{q}} \int_{\min\{\theta,q\}}^{q} \frac{H}{2\sqrt{y}}dy \\
=~& \frac{1}{\sqrt{q}} H [ \sqrt{y} ]_{\min\{\theta,q\}}^{q}  = H \left(1-\sqrt{\frac{\min\{\theta,q\}}{q}}\right)
\end{align*}
If $\theta\leq q$, then this constraint is definitely satisfied, since $b^*\geq 0$. Thus we need to only check the constraints for $q>\theta$. The tightest of these constraints with respect to $\theta$, is when $q=1$, leading to:
\begin{equation}
\theta \geq \left(\frac{H - b^*}{H}\right)^2
\end{equation}
By setting $\theta$ to be equal to the above lower bound, we get a feasible value function and this value function has mean:
\begin{equation*}
\mu^* = H(1-\theta) = H-\frac{(H - b^*)^2}{H} =\frac{H^2 - H^2 -(b^*)^2 + 2b^* H}{H} = 2b^* - \frac{(b^*)^2}{H}
\end{equation*}
which concludes the point of the example.
\end{example}

\subsection{Upper bound on mean of common value}\label{sec:upper-mean}

\begin{theorem}
Let $R=\E_{b\sim D}\left[\max_i b_i\right]$ is the expected revenue of a first price auction when bids are drawn from some BCE and let $\mu=\E[v]$ be the expected common value. If we assume that $v\in [0,H]$ and that the inverse of the CDF of the distribution of values is continuously differentiable in $[0,H]$, then it holds that:
\begin{equation}
\mu\leq \sqrt{\frac{2n}{n-1} H \cdot R}
\end{equation}
If we also assume that the inverse of the CDF of values is $L$-Lipschitz (equivalently the density of values is bounded below by $1/L$), then we don't need to assume an upper bound on the distribution and we alternatively get:
\begin{equation}
\mu\leq \sqrt{\frac{2n}{n-1} \left(\frac{1}{n}R+L\right) \cdot R} = R\sqrt{\frac{2}{n-1}} + \sqrt{\frac{2n}{n-1} L\cdot R}
\end{equation}
\end{theorem}
\begin{proof}
By the main theorem of \cite{Bergemann2016c} we have that if $F$ is the CDF of common values and $H$ is the CDF of the maximum bid distribution at a BCE, while $v(\cdot)=F^{-1}(\cdot)$ and $B(\cdot)=H^{-1}(\cdot)$ are the quantile functions of these distributions, then for any $q\in [0,1]$:
\begin{equation}
B(q)\geq v(q) - \frac{1}{q^{\alpha}}\int_{0}^{v(q)} (F(x))^{\alpha} dx
\end{equation}
with $\alpha = \frac{n-1}{n}$. Let $I(q)$ denote the right hand side.

We  first do a change of variables in the integral: let $x = v(y)$. Then $dx = v'(y)dy$ and since the integral ranges from $x\in [0,v(q)]$, in the new variable the integral ranges from $y\in [0,q]$. Hence:
\begin{equation}
I(q) = v(q) - \frac{1}{q^{\alpha}}\int_{0}^{q} (F(v(y)))^{\alpha} v'(y)dy =
v(q) - \frac{1}{z^\alpha} \int_0^q y^{\alpha} v'(y) dy
\end{equation}
Applying integration by parts on the integral we have:
\begin{align*}
I(q) =~& v(q) - \frac{1}{q^{\alpha}}\left( \left[ y^\alpha v(y)\right]_0^q - \int_{0}^{q} \alpha y^{\alpha-1} v(y)dy \right)\\
=~&  v(q) - v(q) + \frac{1}{q^{\alpha}} \int_{0}^{q} \alpha y^{\alpha-1} v(y)dy\\
=~&  \frac{1}{q^{\alpha}} \int_{0}^{q} \alpha y^{\alpha-1} v(y)dy
\end{align*}

Now observe that, $\mu=\int_{0}^1 v(q) dq$. This can be easily verified pictorially, but also algebraically as follows:
\begin{equation}
\mu = \int_{0}^H 1-F(x) dx = \int_{0}^1 (1 - F(v(q))) v'(q) dq
= [(1-q)v(q)]_0^1 + \int_0^1 v(q) dq  = \int_0^1 v(q) dq
\end{equation}
Similarly, $R= \int_0^1 B(q)dq$. 

Thus integrating the inequality $B(q)\geq I(q)$ over $q$ we get:
\begin{align}
R \geq \int_{0}^1 \frac{1}{q^{\alpha}} \int_{0}^{q} \alpha y^{\alpha-1} v(y)dy dq
\end{align}
Exchanging the integration order, we get:
\begin{align}
R \geq~& \int_{0}^1  \int_{y}^{1} \frac{1}{q^{\alpha}}\alpha y^{\alpha-1} v(y)dq dy = \int_{0}^1 \alpha y^{\alpha-1} v(y) \int_{y}^{1} \frac{1}{q^{\alpha}}dq dy\\
=~& \int_{0}^1 \alpha y^{\alpha-1} v(y) \left[\frac{q^{1-\alpha}}{1-\alpha}\right]_y^1dy\\
=~& \int_{0}^1 \frac{\alpha}{1-\alpha} y^{\alpha-1} v(y) \left(1-y^{1-\alpha}\right)dy\\
=~& \int_{0}^1 \frac{\alpha}{1-\alpha} v(y) \left(y^{\alpha-1}-1\right)dy
\end{align}

Thus by re-arranging we have that:
\begin{equation}
\int_{0}^1 v(y)  \left(y^{\alpha-1}-1\right)dy \leq \frac{1-\alpha}{\alpha} R
\end{equation}

Therefore if we are given as fixed the revenue of the auction $R$, then the mean of the common value distribution can be at most the solution to the following optimization program over the set of all possible quantile functions $v(\cdot)$:
\begin{align}
\max_{v(\cdot)\in [0,H]}& \int_{0}^1 v(y) dy \\
s.t. & \int_0^1 v(y)  \left(y^{\alpha-1}-1\right)dy \leq \frac{1-\alpha}{\alpha} R
\end{align}
Now consider the function $w(y) = y^{\alpha-1} - 1$. This is a monotone decreasing and non-negative function of $y$, starting from $\infty$ at $y=0$ and ending at $0$ at $y=1$. 

The goal of the above linear program is to push as much possible value in $v(y)$ as possible, while satisfying the constraint. Observe that a quantity $v(y)$ is multiplied by a smaller value in the above constraint than any $v(y')$ for $y'\leq y$. Thus it is easy to see that the optimal solution to the above linear program puts as much value on the high $y$'s and sets the remainder $y$'s to zero (another way of arguing this is considering arbitrarily small discretizations of the $y$ space and then arguing that the solution to the linear program with respect to $v(y)$, has the above form). Thus the optimal $v(\cdot)$ takes the form:
\begin{equation}
v(y) = \begin{cases}
0 & \text{ if $y\leq \theta$ }\\
H & \text{ o.w. }
\end{cases}
\end{equation}
for some threshold $\theta\in [0,1]$.

Thus the optimal value of the above program simplifies to:
\begin{align}
\max_{\theta\in [0,1]}&~ H(1-\theta) \\
s.t. & ~ H \int_\theta^1 \left(y^{\alpha-1}-1\right)dy \leq \frac{1-\alpha}{\alpha} R
\end{align}
Thus we want to minimize $\theta$, such that:
\begin{equation}
\frac{1-\alpha}{\alpha} \frac{R}{H} \geq \int_\theta^1 \left(y^{\alpha-1}-1\right)dy  = \frac{1}{\alpha} - \frac{\theta^\alpha}{\alpha}-1+\theta
\end{equation}

\paragraph{Case of $n=2$} For $n=2$, i.e. $\alpha=1/2$, we can exactly solve the latter, since it corresponds to a quadratic inequality, i.e.:
\begin{equation}
\frac{R}{H} \geq 1 - 2\sqrt{\theta}+\theta= (1-\sqrt{\theta})^2
\end{equation}
which yields that: $\sqrt{\theta} \geq 1-\sqrt{\frac{R}{H}}$ or equivalently, $\theta\geq 1+\frac{R}{H} - 2\sqrt{\frac{R}{H}}$. Therefore, the highest possible mean is at most: $H(1-\theta) = H\left(2\sqrt{\frac{R}{H}}-\frac{R}{H}\right)=2 \sqrt{R\cdot H} - R$.

\paragraph{General $n$.} For arbitrary $\alpha\in [0,1]$, we can lower bound the right hand side with a second order Taylor expansion around $\theta=1$: i.e. for any $\theta \in [0,1]$, for some $t\in [\theta, 1]$:
\begin{equation}
\theta^\alpha = 1 + (\theta-1) \alpha + \frac{\alpha (\alpha-1)}{2} t^{\alpha-2} (\theta-1)^2 \leq 1 + (\theta-1) \alpha + \frac{\alpha (\alpha-1)}{2} (\theta-1)^2 
\end{equation}
(since for $\alpha\in [0,1]$, $t\in [0,1]$: $t^{\alpha-2}\geq 1$ and $\alpha \leq 1$)
Replacing the above approximation in the constraint, gives:
\begin{equation}
\frac{1-\alpha}{\alpha} \frac{R}{H} \geq \frac{1}{\alpha} - \frac{1}{\alpha}-\theta+1 - \frac{(\alpha-1)}{2} (\theta-1)^2 -1+\theta = \frac{1-\alpha}{2} (\theta-1)^2
\end{equation}
Simplifying further, we want to minimize $\theta$ such that: $(\theta-1)^2 \leq \frac{2}{\alpha}\frac{R}{H}$. The optimal value sets:
\begin{equation}
\theta = 1 - \sqrt{\frac{2}{\alpha}\frac{R}{H}}
\end{equation}
Leading to a mean of:
\begin{equation}
\mu^* = H ( 1-\theta) \leq \sqrt{\frac{2}{\alpha}\cdot R\cdot H}
\end{equation}

For the second part of the theorem, we simply observe that if the function $v(\cdot)$ is $L$-Lipschitz, then it must be that $v(1) \leq (1-\alpha)R + L$, thereby, we can do the exact same analysis as above, but with $H = (1-\alpha)R+L$. The upper bound on $v(1)$, comes from the fact, that if $v(\cdot)$ is lipschitz, then $v(y) \geq v(1) - L (1-y)\geq v(1)-L$. Thus we have:
\begin{equation}
\frac{1-\alpha}{\alpha} R\geq \int_{0}^1 v(y)  \left(y^{\alpha-1}-1\right)dy \geq (v(1)-L)  \int_{0}^1 \left(y^{\alpha-1}-1\right)dy = \frac{v(1)-L}{\alpha}
\end{equation}
Re-arranging gives the property and concludes the proof of the theorem.
\end{proof}

\subsection{Upper bound on mean of common value via first order conditions} \label{upper bound}
We show an upper bound on the mean of the common value from the observed bid distribution, in the case where bids and their distributions are continuous and admit a density.

Consider a single player $i$. We will look at the auction from the perspective of that player and so we will be dropping the index $i$. Consider one bid $b$ for player $i$ that is in the support of the BCE distribution. It must be that this player is maximizing his utility by submitting $b$ rather than any other bid, when he is recommended to bid $b$. Let $B$ denote the random variable that corresponds to the maximum bid of all players and let $B_i$ be the random variable corresponding to the recommended bid for player $i$. Moreover, let $G_{B}(\cdot | b)$ denote the conditional CDF of the maximum bid, conditional on player $i$ being recommended a bid $b$, and $g_{B}(\cdot|b)$ be the PDF of this distribution. Note that even if the bid distribution of the bidders is continuous, both $G_B$ and $g_b$ are in general discontinuous at $b$: for example, $g_B(b'|b) =0$ for $b' < b$ but $g_B(b|b) = P(B_{-i} \leq b | B = b)$ (where $B_{-i}$ is the maximum bid of the other players), which can be non-zero.  However for $b' > b$, we have that
\begin{equation}\label{eq: cdf}
G_B(b'|b) = P[B_{-i} \leq b' |  B_i = b] = G_{B_{-i}}(b'|b)
\end{equation}
\begin{equation}\label{eq: pdf}
g_B(b'|b) = g_{B_{-i}}(b' | b)
\end{equation}
where $g_{B_{-i}}(\cdot|b)$ is the density of $B_{-i}$ conditional on $B_i = b$ and $G_{B_{-i}}(b'|b)$ the corresponding CDF.
This in particular implies that $G_B(b'|b)$ and $g_B(b'|b)$ are continuous for $b' > b$, and that 
\begin{equation}\label{eq: lim_cdf}
\lim_{b' \to b^+} G_B(b'|b) = G_{B_{-i}}(b'|b)
\end{equation}
\begin{equation}\label{eq: lim_pdf}
\lim_{b' \to b^+} g_B(b'|b) = g_{B_{-i}}(b | b)
\end{equation}
as $G_B$ and $g_B$ are continuous for all $b'$.

Let's consider the expected utility of a player from deviating to a bid $b'$ when being recommended a bid $b$ such that $b' > b$ (this does not apply for $b' < b$):
\begin{equation}
U(b'|b) = \left(\E[v | B\leq b', B_i=b] - b'\right) G_{B}(b|b)
\end{equation}
We get that because $b$ maximizes $U(b'|b)$ for $b' > b$, it must be the case that
\begin{equation}
\lim_{b' \to b^+} \frac{\partial U(b'|b)}{\partial b'} \leq 0
\end{equation}
First we consider the derivative of the quantity $E[v|B\leq b', B_i=b]G_{B}(b'|b)$ for any $b' > b$:
\begin{align*}
E[v|B\leq b', B_i=b]G_{B}(b'|b) =  E\left[v\cdot 1\{B\leq b'\}~|~B_i=b\right]=\int_{0}^{b'} \left(\int_{0}^{\infty} v \cdot f(v|z,b) dv\right)  g_{B}(z|b)dz 
\end{align*}
where $f(v|z,b)$ is the density of the conditional distribution of the common value, conditional on $B=z$ and $B_i=b$. Thus the derivative of the latter with respect to $b'$ is simply:
\begin{align*}
\frac{\partial E[v|B\leq b', B_i=b]G_{B}(b'|b)}{\partial b'} =  \left(\int_{0}^{\infty} v \cdot f(v|b',b) dv\right)  g_{B}(b'|b) = E[v|B=b', B_i=b] g_{B}(b'|b)
\end{align*}
Combining with the derivative of the term $b' G_{B}(b'|b)$ we get:
\begin{equation}
\lim_{b' \to b^+} \big (E[v|B=b', B_i=b] g_{B}(b'|b) - b' g_{B}(b'|b) - G_{B}(b'|b) \big) \leq 0
\end{equation}
Re-arranging yields:
\begin{equation}
\lim_{b' \to b^+}  E[v|B=b', B_i=b] \leq \lim_{b' \to b^+}  b' + \frac{\lim_{b' \to b^+} G_{B}(b'|b)}{\lim_{b' \to b^+} g_{B}(b'|b)}
\end{equation}
Using equations~\eqref{eq: cdf},~\eqref{eq: pdf},~\eqref{eq: lim_cdf} and~\eqref{eq: lim_pdf}, we obtain:
\begin{equation}\label{eq: cond_exp}
\lim_{b' \to b^+}  E[v|B=b', B_i=b] \leq b + \frac{G_{B_{-i}}(b|b)}{g_{B_{-i}}(b | b)}
\end{equation}
Note that assuming continuity of $f(v|b',b)$ (i.e. the joint distribution of realized value and bids is continuous), we have continuity of  $E[v| B = b', B_i = b] = E[v | B_{-i} \leq b', B_i =b]$ for all $b'$ and we can write:
\begin{equation}\label{eq: cond_exp}
E[v|B=b, B_i=b] \leq b + \frac{G_{B_{-i}}(b|b)}{g_{B_{-i}}(b | b)} 
\end{equation}

Assuming any two players never play the same bid at the same time with non-zero probability (this happens as long as $g_{B_{-i}}(\cdot | b)$ is not degenerate for all bids $b$ and all players $i$), we have that for all $b$, the events $(B=b,B_i=b)$ are disjoint across possible values of $i$ and $b$ (only one can happen at a time). Additionally, there always exist a player $i$ such that $B=b, B_i=b$ for some bid $b$, hence at least one of the events must happen. So exactly one of $(B = b, B_i = b)$ happens at a time for all possible bids $b$ and bidders $i$, and we can condition on said events:
\begin{equation}
E[v] = \int\limits_b \sum\limits_i E[v|B = b, B_i = b]  P[B = b , B_i = b] db
\end{equation}
Plugging in Equation~\eqref{eq: cond_exp}, we get:
\begin{align*}
E[v] & \leq \sum\limits_i \int\limits_b  \big(b +  \frac{G_{B_{-i}}(b|b)}{g_{B_{-i}}(b | b)}  \big) \cdot G^i_{B}(b|b) g_{B_i}(b) db
\\& = R +  \sum\limits_i \int\limits_b  \frac{G^i_{B}(b|b)^2}{g_{B_{-i}}(b | b)}  \cdot g_{B_i}(b) db
\end{align*}

%% file: bce_econometrics_auctions.bbl
\begin{thebibliography}{18}
\providecommand{\natexlab}[1]{#1}
\providecommand{\url}[1]{\texttt{#1}}
\expandafter\ifx\csname urlstyle\endcsname\relax
  \providecommand{\doi}[1]{doi: #1}\else
  \providecommand{\doi}{doi: \begingroup \urlstyle{rm}\Url}\fi

\bibitem[Audibert et~al.(2009)Audibert, Munos, and Szepesvári]{Audibert2009}
J.-Y. Audibert, R.~Munos, and C.~Szepesvári.
\newblock Exploration–exploitation tradeoff using variance estimates in
  multi-armed bandits.
\newblock \emph{Theoretical Computer Science}, 410\penalty0 (19):\penalty0 1876
  -- 1902, 2009.
\newblock ISSN 0304-3975.
\newblock \doi{https://doi.org/10.1016/j.tcs.2009.01.016}.
\newblock URL
  \url{http://www.sciencedirect.com/science/article/pii/S030439750900067X}.
\newblock Algorithmic Learning Theory.

\bibitem[Aumann(1987)]{aumann1987}
R.~J. Aumann.
\newblock Correlated equilibrium as an expression of bayesian rationality.
\newblock \emph{Econometrica: Journal of the Econometric Society}, pages 1--18,
  1987.

\bibitem[Bergemann and Morris(2013)]{Bergemann2013}
D.~Bergemann and S.~Morris.
\newblock Robust predictions in games with incomplete information.
\newblock \emph{Econometrica}, 81\penalty0 (4):\penalty0 1251--1308, 2013.
\newblock ISSN 1468-0262.
\newblock \doi{10.3982/ECTA11105}.
\newblock URL \url{http://dx.doi.org/10.3982/ECTA11105}.

\bibitem[Bergemann and Morris(2016)]{Bergemann2016}
D.~Bergemann and S.~Morris.
\newblock Bayes correlated equilibrium and the comparison of information
  structures in games.
\newblock \emph{Theoretical Economics}, 11\penalty0 (2):\penalty0 487--522,
  2016.
\newblock URL
  \url{https://econtheory.org/ojs/index.php/te/article/view/20160487/0}.

\bibitem[Bergemann et~al.(2017)Bergemann, Brooks, and Morris]{Bergemann2016c}
D.~Bergemann, B.~Brooks, and S.~Morris.
\newblock First-price auctions with general information structures:
  Implications for bidding and revenue.
\newblock \emph{Econometrica}, 85\penalty0 (1):\penalty0 107--143, 2017.
\newblock ISSN 1468-0262.
\newblock \doi{10.3982/ECTA13958}.
\newblock URL \url{http://dx.doi.org/10.3982/ECTA13958}.

\bibitem[Chamberlain and Imbens(2003)]{ChamberlainImbens}
G.~Chamberlain and G.~W. Imbens.
\newblock Nonparametric applications of bayesian inference.
\newblock \emph{Journal of Business \& Economic Statistics}, 21\penalty0
  (1):\penalty0 12--18, 2003.

\bibitem[Chernozhukov et~al.(2007)Chernozhukov, Hong, and Tamer]{cht}
V.~Chernozhukov, H.~Hong, and E.~Tamer.
\newblock Estimation and confidence regions for parameter sets in econometric
  models.
\newblock \emph{Econometrica}, 75\penalty0 (5):\penalty0 1243--1284, 2007.

\bibitem[Guerre et~al.(2000)Guerre, Perrigne, and Vuong]{guerre2000}
E.~Guerre, I.~Perrigne, and Q.~Vuong.
\newblock Optimal nonparametric estimation of first-price auctions.
\newblock \emph{Econometrica}, 68\penalty0 (3):\penalty0 525--574, 2000.

\bibitem[Haile and Tamer(2003)]{HaileTamer}
P.~A. Haile and E.~Tamer.
\newblock Inference with an incomplete model of english auctions.
\newblock \emph{Journal of Political Economy}, 111\penalty0 (1):\penalty0
  1--51, 2003.

\bibitem[Hendricks and Porter(1988)]{hendricksporter}
K.~Hendricks and R.~H. Porter.
\newblock An empirical study of an auction with asymmetric information.
\newblock \emph{The American Economic Review}, pages 865--883, 1988.

\bibitem[Hoy et~al.(2015)Hoy, Nekipelov, and Syrgkanis]{Hoy2015}
D.~Hoy, D.~Nekipelov, and V.~Syrgkanis.
\newblock Robust data-driven efficiency guarantees in auctions.
\newblock \emph{CoRR}, abs/1505.00437, 2015.
\newblock URL \url{http://arxiv.org/abs/1505.00437}.

\bibitem[Kline and Tamer(2016)]{KlineTamer}
B.~Kline and E.~Tamer.
\newblock Bayesian inference in a class of partially identified models.
\newblock \emph{Quantitative Economics}, 7\penalty0 (2):\penalty0 329--366,
  2016.

\bibitem[Li et~al.(2002)Li, Perrigne, and Vuong]{Li2002}
T.~Li, I.~Perrigne, and Q.~Vuong.
\newblock Structural estimation of the affiliated private value auction model.
\newblock \emph{The RAND Journal of Economics}, 33\penalty0 (2):\penalty0
  171--193, 2002.
\newblock ISSN 07416261.
\newblock URL \url{http://www.jstor.org/stable/3087429}.

\bibitem[Magnolfi and Roncoroni(2016)]{Magnolfi2016}
L.~Magnolfi and C.~Roncoroni.
\newblock Estimation of discrete games with weak assumptions on information.
\newblock Technical report, mimeo, 2016.

\bibitem[{Maurer} and {Pontil}(2009)]{Mauer2009}
A.~{Maurer} and M.~{Pontil}.
\newblock {Empirical Bernstein Bounds and Sample Variance Penalization}.
\newblock \emph{ArXiv e-prints}, July 2009.

\bibitem[Peel et~al.(2010)Peel, Anthoine, and Ralaivola]{Peel2010}
T.~Peel, S.~Anthoine, and L.~Ralaivola.
\newblock Empirical bernstein inequalities for u-statistics.
\newblock In \emph{NIPS}, 2010.

\bibitem[Shapiro(2009)]{Shapiro2009ch5}
A.~Shapiro.
\newblock \emph{5. Statistical Inference}, pages 155--252.
\newblock 2009.
\newblock \doi{10.1137/1.9780898718751.ch5}.
\newblock URL \url{http://epubs.siam.org/doi/abs/10.1137/1.9780898718751.ch5}.

\bibitem[Syrgkanis and Tardos(2013)]{Syrgkanis2013}
V.~Syrgkanis and E.~Tardos.
\newblock Composable and efficient mechanisms.
\newblock In \emph{Proceedings of the Forty-fifth Annual ACM Symposium on
  Theory of Computing}, STOC '13, pages 211--220, New York, NY, USA, 2013. ACM.
\newblock ISBN 978-1-4503-2029-0.
\newblock \doi{10.1145/2488608.2488635}.
\newblock URL \url{http://doi.acm.org/10.1145/2488608.2488635}.

\end{thebibliography}
